\documentclass[a4paper, 12pt]{article}

\usepackage{epsfig}
\usepackage{subfigure}
\usepackage{calc}
\usepackage{amssymb}
\usepackage{amstext}
\usepackage{amsmath}
\usepackage{multicol}
\usepackage{pslatex}
\usepackage{latexsym}
\usepackage[all]{xy}
\usepackage{array}
\usepackage{tikz}
\usetikzlibrary{trees}
\usepackage{ifsym}
\usepackage{epsf}
\usepackage{verbatim}
\usepackage{amsfonts,amsmath,latexsym}
\usepackage{bbm}
\usepackage{etex}
\usepackage{graphicx}
\usepackage{cite}
\usepackage{array}
\usetikzlibrary{positioning}
\usepackage{slashbox}
\usepackage{fancyvrb}
\usepackage{epsf}
\usepackage{pgfplots}

\newcommand{\At}{\mathrm{\textrm{At}}}

\newcommand{\ls}{\mathcal{L}_{\Sigma}}

\begin{document}


\newcounter{thm}
\newcounter{lem}
\newcounter{def}
\newtheorem{definition}[thm]{Definition}
\newtheorem{theorem}[thm]{Theorem}
\newtheorem{lemma}[thm]{Lemma}
\newtheorem{corollary}[thm]{Corollary}
\newtheorem{example}[thm]{Example}
\newtheorem{proposition}[thm]{Proposition}

\newcommand{\proof}{\noindent {\bf Proof }}

\author{Ekaterina Komendantskaya\footnote{The work was supported by the Engineering and Physical Sciences Research Council, UK: Postdoctoral Fellow in TCS grant EP/F044046/1-2, EPSRC First Grant EP/J014222/1, and EPSRC Grant EP/K031864/1.}\\
  School of Computing,\\ University of Dundee, UK \and  John Power\footnote{The work was supported by  Royal Society grant "Universal Algebra and its dual: monads and comonads, Lawvere theories and what?", EPSRC grant EP/K028243/1, and SICSA Distinguished Visiting Fellow grant.} \\ Department of Computer Science,\\
University of Bath, UK \and Martin Schmidt \\ Institute of Cognitive Science,\\ Osnabr\"{u}ck University, Germany\\ }
\date{}

\title{Coalgebraic Logic Programming:\\ from Semantics to Implementation}
\maketitle 
\abstract{Coinductive definitions, such as that of an infinite stream, may often
  be described by elegant logic programs, but ones for which
  SLD-refutation is of no value as SLD-derivations fall into infinite loops. 
	Such definitions give rise to questions of lazy corecursive derivations and 
  parallelism, as execution of such logic programs can have both recursive and
  corecursive features at once. Observational and coalgebraic
  semantics have been used to study them abstractly. The programming developments
  have often occurred separately and have usually been
  implementation-led. Here, we give a coherent semantics-led account
  of the issues, starting with abstract category theoretic semantics,
  developing coalgebra to characterise naturally arising trees, and
  proceeding towards implementation of a new dialect, CoALP, of logic
  programming, characterised by guarded lazy corecursion and parallelism.}




\section{Introduction}

The central algorithm of logic programming is
$SLD$-resolution~\cite{Kow79,Llo88,SS86}.  It is primarily used to
obtain $SLD$-refutations; it is usually given least fixed point
semantics; and it is typically implemented
sequentially~\cite{Llo88,SS86}.

All three of these traditions have been challenged over the years, for
related reasons.  For example, infinite streams of bits can be described naturally in
terms of a logic program \texttt{Stream}:
\vspace{-5pt}
\small{
    \begin{eqnarray*}
        \texttt{bit}(0) & \gets & \\
        \texttt{bit}(1) & \gets & \\
        \texttt{stream(scons (x,y))} & \gets & \texttt{bit(x)},\texttt{stream(y)}\\
    \end{eqnarray*}
  } 
$SLD$-resolution is of value here, but $SLD$-refutations are not,
  and that is standard for coinductively defined
  structures~\cite{GuptaBMSM07,Jaume02,Llo88,Majki?04coalgebraicsemantics,
    SimonBMG07}. Consequently, least fixed point semantics, which is
  based on finiteness of derivations, is unhelpful. \texttt{Stream}
  can be given greatest fixed point semantics~\cite{Llo88}, but
  greatest fixed point semantics is incomplete in general, failing for
  some infinite derivations.  \texttt{Stream} can alternatively be
  given coalgebraic semantics~\cite{BonchiM09,CominiLM01} or
  observational semantics~\cite{CominiLM01,GLM95}. Coalgebraic
  semantics is, in general, well-suited to describing parallel
  processes~\cite{JR97,R2000}.

  In this paper, we propose a single coherent, conceptual
  semantics-led framework for this, developing and extending three
  recent conference papers~\cite{KMP10,KP11-2,KP11}. We start from the
  theoretical, with an abstract category theoretic semantics for logic
  programming, and we proceed to the applied, ultimately proposing a
  new dialect, CoALP, of logic programming based on our abstract
  development. We do not change the definition of a logic program; we
  rather change the analysis of it. \texttt{Stream} is a leading and
  running example for us.

In more detail, a first-order logic program consists of a finite set
of clauses of the form
\[
A \leftarrow A_1, \ldots , A_n
\]
where $A$ and the $A_i$'s are atomic formulae, typically containing
free variables, and with the $A_i$'s mutually distinct.  In the ground
case, i.e., if there are no free variables, such a logic program can
be identified with a function $p:At\longrightarrow P_f(P_f(At))$,
where $At$ is the set of atomic formulae and $p$ sends an atomic
formula $A$ to the set of sets of atomic formulae in each antecedent
of each clause for which $A$ is the
head~\cite{BonchiM09,CominiLM01,JR97}. Such a function is called a
coalgebra for the endofunctor $P_fP_f$ on the category $Set$.  Letting
$C(P_fP_f)$ denote the cofree comonad on $P_fP_f$, given a ground
logic program qua $P_fP_f$-coalgebra, we characterise
and-or parallel derivation trees~\cite{GuptaC94,PontelliG95,GPACH12}
in terms of the $C(P_fP_f)$-coalgebra structure corresponding to $p$,
see Section~\ref{sec:OS}. And-or parallel derivation trees subsume
$SLD$-trees and support parallel implementation and the {\em Theory of
  Observables}~\cite{GLM95,CominiLM01}.

The extension from ground logic programs to first-order programs is
subtle, requiring new abstract category theory. Nevertheless, it
remains in the spirit of the situation for ground logic programs.  Our
characterisation of  and-or parallel derivation trees
does not extend from ground to arbitrary logic programs exactly, but
it fails in particularly interesting ways: the relationship between
and-or trees and ours is at the heart of the paper. Indeed, the
analysis of trees is fundamental to us. We end our abstract
development by proving soundness, completeness, correctness and full
abstraction results for coalgebraic semantics in Section~\ref{sec:OS}.

Proceeding from the abstract to the applied, two aspects of logic
programming that are both desirable and problematic in practice are
corecursion and parallelism. 

Many accounts of corecursion in logic programming, e.g., CoLP~\cite{GuptaBMSM07,SimonBMG07}, use
explicit annotation of corecursive loops to terminate infinite
derivations, see Section~\ref{sec:corec}. In such accounts, inductive
and coinductive predicates are labelled in order to make the
distinction between admissible (in corecursion) and non-admissible (in
recursion) infinite loops. But some predicates need
to be treated as recursive or corecursive depending on the context,
making annotation prior to program execution
impossible. Example~\ref{ex:stream2}, extending \texttt{Stream},
illustrates this.

We propose an alternative approach to corecursion in logic programming:
 a new derivation
algorithm based on the coinductive trees -- structures directly inspired by our coalgebraic semantics.
The resulting dialect  CoALP is based on the same syntax of Horn-clause logic programming, but, in place of SLD-resolution, it  features a new \emph{coinductive derivation} algorithm. CoALP\rq{}s lazy corecursive derivations and syntactic \emph{guardedness} rules are 
similar to those implemented in lazy functional languages, cf. \cite{BK08,Coq94,Gimenez98}.
Unlike alternative approaches~\cite{GuptaBMSM07,SimonBMG07}, CoALP does not require explicit syntactic annotations of coinductive definitions.
We discuss coinductive trees and derivations in Section~\ref{sec:corec}. There, we prove soundness and completeness of CoALP relative to the coalgebraic semantics of Section \ref{sec:OS}.

Another distinguishing feature of logic programming languages is that
they allow implicit parallel execution of programs. The three main
types of parallelism used in implementations are
\emph{and-parallelism}, \emph{or-parallelism}, and their combination:
\cite{GuptaC94,PontelliG95,GPACH12}.  However, many first-order
algorithms are P-complete and hence inherently
sequential~\cite{DKM84,Kanellakis88}. This especially concerns
first-order unification and variable substitution in the presence of
variable dependencies. Care is required here. For example, in
\texttt{Stream}, the goal \texttt{stream(scons(x, scons(y,x)))}, if
processed sequentially, leads to a failed derivation owing to
ill-typing, whereas if proof search proceeds in a parallel fashion, it
may find substitutions for $x$, e.g., $0$ and $nil$, in distinct
parallel branches of the derivation tree, but such a derivation is not
sound, see Example~\ref{ex:unsound}.

Existing implementations \cite{GuptaC94,PontelliG95,GPACH12} of
parallel $SLD$-derivations require keeping records of previous
substitutions and so involve additional data structures and algorithms
that coordinate variable substitution in different branches of
parallel derivation trees; which ultimately restricts parallelism.  
If such synchronisation is omitted,
parallel $SLD$-derivations may lead to unsound results as in
\texttt{Stream} above.  Again, this can be seen as explicit resource
handling, where resources are variables, terms, and substitutions.  In
Kowalski's terms of \emph{Logic Programming = Logic + Control}
\cite{Kow79}, this leads to the separation of issues of logic
(unification and $SLD$-resolution) and control (underlying
implementation tools) in most parallel logic programming
implementations, as we explain in Section \ref{sec:parallel}.

CoALP offers an alternative solution to this problem.  The coinductive resolution of
CoALP has an inherent ability to handle parallelism. 
Namely, coinductive trees with imposed guardedness conditions provide a natural formalism for
parallel implementation of coinductive derivations. Parallelisation of CoALP is sound by (guarded) program
construction and the construction of coinductive trees.
The main
distinguishing features of parallelism in CoALP are implicit resource
handling and convergence of the issues of logic and control:
no explicit scheduling of parallel processes is needed, and parallelisation is handled by the coinductive derivation algorithm.
We explain this in Section \ref{sec:parallel}.  

Ultimately, in Section~\ref{sec:impl},  we propose the first implementation of CoALP, available for download from \cite{SK12}.  Its main
distinguishing features are guarded corecursion, parallelism, and implicit handling of corecursive and parallel resources.
In Section \ref{sec:concl} we conclude and discuss future work.

\section{SLD Derivations and Trees they
  Generate}\label{sec:lp1}\label{sec:FOLP}

We recall the definitions surrounding the notion
of $SLD$-derivation~\cite{Llo88}, and we consider
various kinds of trees the notion generates.

\subsection{Background Definitions}

\begin{definition}\label{df:syntax}
  A \emph{signature} $\Sigma$ consists of a set of \emph{function
    symbols} $f,g, \ldots$ each equipped with an
  \emph{arity}. The arity of a function symbol is a natural number
  indicating the number of arguments it has.  Nullary
  (0-ary) function symbols are called \emph{constants}.
\end{definition}
 
Given a countably infinite set $Var$ of variables, denoted $x,y,z$,
sometimes with indices $x_1, x_2, x_3, \ldots$, terms are defined as
follows.
\begin{definition}
The set $Ter(\Sigma)$ of \emph{terms} over $\Sigma$ is defined
inductively:
\begin{itemize}
\item $x \in Ter(\Sigma)$ for every $x \in Var$.
\item If $f$ is an n-ary function symbol and $t_1,\ldots
  ,t_n \in Ter(\Sigma) $, then $f(t_1,\ldots
  ,t_n) \in Ter(\Sigma)$.  
\end{itemize}  
\end{definition}


\begin{definition}\label{df:subst}
A \emph{substitution} is a function $\theta : Ter(\Sigma) \rightarrow
Ter(\Sigma) $ which satisfies
$$ \theta (f(t_1,\ldots
  ,t_n)) = f(\theta (t_1),\ldots
  ,\theta(t_n))$$
for every n-ary function symbol $f$.
\end{definition}

An \emph{alphabet} consists of a signature $\Sigma$, the set $Var$,
and a set of \emph{predicate symbols} $P_1, P_2, \ldots$, each
assigned an arity.  If $P$ is a predicate symbol of arity $n$ and
$t_1, \ldots, t_n$ are terms, then $P(t_1, \ldots, t_n)$ is a
\emph{formula}, also called an {\em atomic formula} or an
\emph{atom}.
The \emph{first-order language $\mathcal{L}$} given by an alphabet consists of
the
set of all formulae constructed from the symbols of the alphabet.

\begin{definition}\label{df:language}
  Given a first-order language $\mathcal{L}$, a \emph{logic program}
  consists of a finite set of {\em clauses} of the form $A \gets A_1,
  \ldots , A_n,$ where $A$ is an atom and $A_1, \ldots A_n$ ($n\geq 0$)
  are distinct atoms.  The atom $A$ is called the \emph{head} of the
  clause, and $A_1, \ldots, A_n$ is called its \emph{body}.  Clauses
  with empty bodies are called \emph{unit clauses}.
A \emph{goal} is given by $\gets A_1, \ldots A_n$, where $A_1, \ldots
A_n$ ($n \geq 0$) are distinct atoms.
\end{definition}

Logic programs of Definition \ref{df:language} are also called Horn-clause logic programs \cite{Llo88}.

\begin{example}\label{ex:stream}
  Program \texttt{Stream} from Introduction defines infinite streams of binary bits. Its
  signature consists of two constants, $0$ and $1$, and a binary
  function symbol \texttt{scons}. It involves two predicate symbols,
  \texttt{bit} and \texttt{stream}, and it has five atoms, arranged
  into three clauses, two of which are unit clauses.  The body of the
  last clause contains two atoms. 
\end{example}

\begin{example}\label{ex:listnat}
\texttt{ListNat}
 denotes the logic program
\begin{eqnarray*}
\texttt{nat(0)} & \gets &\\
\texttt{nat(s(x))} & \gets & \texttt{nat(x)}\\
\texttt{list(nil)} & \gets & \\
\texttt{list(cons ( x, y ))} & \gets & \texttt{nat(x), list(y)}\\
\end{eqnarray*}

\end{example}

Operational semantics for logic programs is given by $SLD$-resolution,
a goal-oriented proof-search procedure.  

\begin{definition}\label{df:mgu}
  Let $S$ be a finite set of atoms. A substitution $\theta$ is called
  a \emph{unifier} for $S$ if, for any pair of atoms $A_1$ and $A_2$
  in $S$, applying the substitution $\theta$ yields $A_1\theta =
  A_2\theta$.  A unifier $\theta$ for $S$ is called a \emph{most
    general unifier} (mgu) for $S$ if, for each unifier $\sigma$ of
  $S$, there exists a substitution $\gamma$ such that $\sigma =
  \theta\gamma$. If $\theta$ is an mgu for $A_1$ and $A_2$, moreover, $A_1\theta = A_2$, then $\theta$ is a \emph{term-matcher}.
\end{definition}

We assume that, given a goal $G =\;  \gets B_1, \ldots , B_n$, there is an
algorithm that, given $B_1, \ldots , B_n$, outputs $B_i$, $i \in \{1,
\ldots , n\}$. The resulting atom $B_i$ is called the \emph{selected
  atom}. Most PROLOG implementations use the algorithm that selects
the left-most atom in the list $B_1, \ldots , B_n$ and proceeds inductively.

\begin{definition}\label{df:SLD}
Let a goal $G$ be $\gets A_1,\ldots ,A_m, \ldots, A_k$ and a clause $C$ be
$A\gets B_1, \ldots ,B_q$. Then $G'$ is \emph{derived} from $G$ and $C$ using
mgu
$\theta$ if the following conditions hold:
\item[$\bullet$] $\theta$ is an mgu of the selected
atom $A_m$  in $G$ and $A$;
\item[$\bullet$] $G'$ is the goal $\gets(A_1, \ldots, A_{m-1},B_1, \ldots
,B_q, A_{m+1},
\ldots, A_k)\theta$.
\end{definition}

A clause $C^*_i$ is a \emph{variant} of the clause $C_i$ if $C^*_i =
C_i \theta$, with $\theta$ being a variable renaming substitution such
that variables in $C_i^*$ do not appear in the derivation up to
$G_{i-1}$. This process of renaming variables is called
\emph{standardising the variables apart}; we assume it throughout the
paper without explicit mention.

\begin{definition}\label{df:SLDderiv}
  An \emph{SLD-derivation} of $P\cup \{G\}$ consists of a sequence of
  goals $G=G_0, G_1, \ldots$ called \emph{resolvents}, a sequence
  $C_1,C_2, \ldots$ of variants of program clauses of $P$, and a
  sequence $\theta_1,\theta_2,\ldots$ of mgu's such that each $G_{i+1}$
  is derived from $G_i$ and $C_{i+1}$ using $\theta_{i+1}$. An
  \emph{SLD-refutation} of $P \cup \{G\}$ is a finite $SLD$-derivation
  of $P \cup \{G\}$ for which the last goal $G_n$ is empty, denoted by
  $\Box$. If $G_n = \Box$, we say that the refutation has length
  $n$. The composite $\theta_1 \theta_2, \ldots$ is called a
  \emph{computed answer}.
\end{definition}

Traditionally, logic programming has been modelled by {\em least fixed
  point} semantics~\cite{Llo88}. Given a logic program $P$, one lets
$B_P$ (also called a \emph{Herbrand base}) denote the set of atomic
ground formulae generated by the syntax of $P$, and one defines
$T_P(I)$ on  $2^{B_P}$ by sending $I$ to the set $\{A \in B_P : A
\gets A_1,...,A_n$ is a ground instance of a clause in $P$ with
$\{A_1, ... ,A_n\} \subseteq I\}$.  The least fixed point of $T_P$ is
called the {\em least Herbrand model} of $P$ and duly satisfies
model-theoretic properties that justify that expression \cite{Llo88}.

SLD-resolution is sound and complete with respect to least
fixed point semantics~\cite{Llo88}.  
The classical theorems of soundness and completeness of this operational semantics \cite{Llo88,FalaschiLPM89,FalaschiLMP93} 
show that every atom in the set computed by the least fixed point of $T_P$ has a finite SLD-refutation, and vice versa.
Alternatively, in \cite{KP96,KP08}, we described an algebraic (fibrational) semantics
for logic programming and proved soundness and completeness results
for it with respect to SLD-resolution. Other forms of algebraic
semantics for logic programming have been given in
\cite{AmatoLM09,BruniMR01}. See also Figure \ref{pic:sem}.

However, Programs like \texttt{Stream} induce infinite SLD-derivations and require a greatest fixed point semantics. 
The greatest fixed point semantics for SLD derivations yields soundness, but not completeness results.

\begin{example}\label{ex:lp2}
The program \texttt{Stream}
is characterised by the greatest fixed point of
the $T_P$ operator, which contains $\texttt{stream}(scons^{\omega}(X,Y))$; 
whereas no infinite term can be computed via SLD-resolution.  
\end{example}

\begin{example}\label{ex:lp3}
For the program $R(x)  \gets  R(f(x))$, 
the greatest fixed point of
the $T_P$ operator contains $R(f^{\omega}(a))$,
but no infinite term is computed by SLD-resolution.  
\end{example}

There have been numerous attempts to resolve the mismatch between infinite derivations and greatest fixed point semantics 
\cite{GuptaBMSM07,Jaume02,Llo88,Majki?04coalgebraicsemantics,SimonBMG07}. 
Here, extending~\cite{KP11,KP11-2}, we give a uniform semantics of infinite SLD derivations for both finite and infinite objects, see Figure \ref{pic:sem}.
Coalgebraic semantics has been used to model various aspects of programming \cite{JR97,Milner89,R2000}, in particular, logic programming~\cite{BonchiM09,CominiLM01};
here, we use it to remedy incompleteness for corecursion.

\begin{figure}
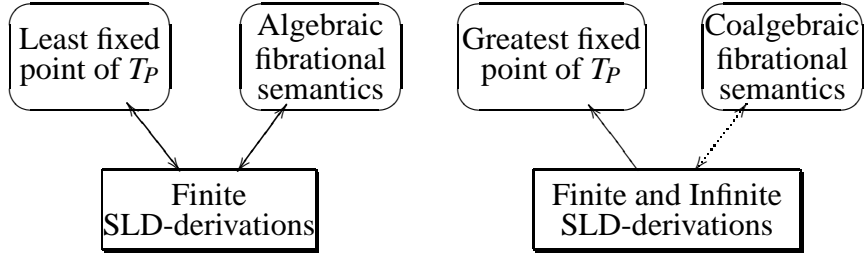

$$
\xy0;/r.12pc/: 
*[o]=<80pt,30pt>\hbox{\txt{Finite\\ SLD-derivations}}="a"*\frm<8pt>{-,},
(-30,40)*[o]=<60pt,40pt>\hbox{\txt{Least fixed\\ point of $T_P$}}="r"*\frm<8pt>{-},
(30,40)*[o]=<60pt,40pt>\hbox{\txt{Algebraic\\ fibrational\\ semantics}}="b"*\frm<8pt>{-},
(120,0)*[o]=<100pt,30pt>\hbox{\txt{Finite and Infinite\\SLD-derivations}}="c"*\frm<8pt>{-,},
(90,40)*[o]=<70pt,40pt>\hbox{\txt{Greatest fixed\\ point of $T_P$}}="d"*\frm<8pt>{-},
(150,40)*[o]=<60pt,40pt>\hbox{\txt{Coalgebraic\\ fibrational\\ semantics}}="e"*\frm<8pt>{-},
"a";"b" **\dir{-} ?>*\dir{>},
"b";"a" **\dir{-} ?>*\dir{>},
"a";"r" **\dir{-} ?>*\dir{>},
"r";"a" **\dir{-} ?>*\dir{>},
"c";"d" **\dir{-} ?>*\dir{>},
"c";"e" **\dir{.} ?>*\dir{>},
"e";"c" **\dir{.} ?>*\dir{>}
\endxy
$$
\caption{\footnotesize{Alternative semantics for finite and infinite SLD-derivations. 
The arrows $\leftrightarrow$ show the semantics that are both sound and complete, and the arrow $\rightarrow$
indicates sound incomplete semantics. The dotted arrow indicates the sound and complete semantics we propose here.}}\label{pic:sem}
\end{figure}

\subsection{Tree Structures in Analysis of Derivations}\label{sec:trees2}

Coalgebraic Logic Programming (CoALP) we introduce in later sections uses a variety of tree-structures both for giving semantics to logic programming  and
for  implementation of CoALP. Here, we briefly survey the kinds of trees traditionally used in logic programming.

For a given goal $G$, there may be several possible $SLD$-derivations
as there may be several clauses with the same head. The definition of
\emph{SLD-tree} allows for this as follows.
 

\begin{definition}\label{df:SLDtree}
  Let $P$ be a logic program and $G$ be a goal. An \emph{SLD-tree} for
  $P\cup \{G\}$ is a possibly infinite tree $T$ satisfying the
  following:
\begin{enumerate}
\item\label{it:2} the root node is $G$
\item\label{it:1} each node of the tree is a (possibly empty) goal
\item\label{it:3} if $\gets A_1,\ldots , A_m$, $m>0$ is a node in $T$,
  and it has $n$ children, then there exists $A_k \in A_1,\ldots ,
  A_m$ such that $A_k$ is unifiable with exactly $n$ distinct clauses
  $C_1 = A^1 \gets B^1_1, \ldots , B^1_q$, ..., $C_n = A^n \gets
  B^n_1, \ldots , B^n_r$ in $P$ via mgu's $\theta_1, \ldots \theta_n$,
  and, for every $i \in \{1, \ldots n\}$, the $i$th child node is
  given by the goal
$$\gets (A_1,\ldots , A_{k-1}, B^i_1, \ldots , B^i_q, A_{k+1}, \ldots , A_m)\theta_i$$  
\item\label{it:4} nodes which are the empty clause have no children.
\end{enumerate}
\end{definition} 

Each branch of an $SLD$-tree is an $SLD$-derivation of $P\cup
\{G\}$. Branches corresponding to successful derivations are called
\emph{success branches}, branches corresponding to infinite
derivations are called \emph{infinite branches}, and branches
corresponding to failed derivations are called \emph{failure
  branches}. A distinctive feature of the SLD-trees is that they allow to exploit
alternative choices of clauses in the proof-search; for this reason, they are also known as \emph{or-trees}.	
See Figure \ref{fig:SLD1}.

In parallel logic programming~\cite{GPACH12}, or-parallelism is exploited when more than one clause unifies with
the goal. It is thus a way of efficiently searching for
solutions to a goal, by exploring alternative solutions in parallel.
It has been implemented in Aurora~\cite{Lusk90} 
and Muse~\cite{AK91}, 
both
of which have shown good speed-up results over a considerable range of
applications.

	Each $SLD$-derivation, or, equivalently, each branch of
an $SLD$-tree, can be represented by a proof-tree, defined as follows.
 
 \begin{definition}\label{df:prooftree}
   Let $P$ be a logic program and $G =\; \gets A$ be an atomic goal. A
   \emph{proof-tree} for $A$ is a possibly infinite tree $T$ such that
\begin{itemize}
\item $A$ is the root of $T$.
\item Each node in $T$ is an atom.
\item For every node $A'$ occurring in $T$, if $A'$ has children $C_1,
  \ldots, C_m$, then there exists a clause $B \gets B_1, \ldots, B_m$
  in $P$ such that $B$ and $A'$ are unifiable with mgu $\theta$, and
  $B_1\theta = C_1$, ... ,$B_m\theta = C_m$.
\end{itemize}
\end{definition} 

Proof-trees exploit the branching occurring when one constructs derivations for several atoms in a goal; and are also known as \emph{and-trees}. 
In parallel logic programming, and-parallelism arises when more than one atom is present in
the goal.  That is, given a goal $G \ = \ \gets B_1, \ldots B_n$, an
and-parallel algorithm for $SLD$-resolution looks for
$SLD$-derivations for each $B_i$ simultaneously, subject to the
condition that the atoms must not share variables.  Such cases are
known as independent and-parallelism.
Independent and-parallelism has been successfully exploited in
$\&$-PROLOG~\cite{HermenegildoG90}.

\begin{example}
Figure \ref{fig:SLD1} depicts a proof tree and an SLD-tree for the
goal \texttt{list(x)} in \texttt{ListNat}. 
\end{example}

\begin{example}\label{ex:stream1}
  \texttt{Stream}, i.e., Example~\ref{ex:stream}, allows the following
  infinite $SLD$-derivation
  \\
  \texttt{stream(x)} $\xrightarrow{x/scons(y,z)}$
  \texttt{bit(y),stream(z)} $\xrightarrow{y/0}$ \texttt{stream(z)}
  $\rightarrow$ $\ldots$
  \\
  containing an infinite repetition of \texttt{stream(x)} for various
  variables $x$. So \texttt{Stream} gives rise to infinite
  $SLD$-trees.
	\end{example}


\begin{figure}

\begin{center}
  \begin{tikzpicture}[baseline=(current bounding box.north),grow=down,level distance=10mm,sibling distance=15mm,font=\footnotesize]
  \node {$\mathtt{list(x)}$}
    		child { node {$\mathtt{nat(y)}$}
       	child { node{$\Box$}
				node[left] {$\theta_0\theta_1$}
				}
				node[right] {} }
     child { node {$\mathtt{list(z)}$}
       child { node{$\Box$}
				node[left] {$\theta_0\theta_2$}}
       }; 
  \end{tikzpicture}
\quad\quad\quad
\begin{tikzpicture}[baseline=(current bounding box.north),grow=down,level distance=8mm,sibling distance=20mm,font=\footnotesize]
    \node {$\mathtt{list(x)}$}
		child [sibling distance=40mm]{ node {$\Box$}
		edge from parent
				node[left] {$\theta_4$}}
      child [sibling distance=40mm]{ node {$\mathtt{nat(y),list(z)}$}
        child  [sibling distance=60mm] { node {$\mathtt{list(z)}$}
          child  [sibling distance=20mm]{ node{$\mathtt{\Box}$}
					edge from parent
				node[left] {$\theta_2$}}
					child [sibling distance=20mm] { node{$\mathtt{list(z1)}$}
					   child  [sibling distance=10mm] {node{$\Box$}}
					   child  [sibling distance=10mm]{node{$\vdots$}}}
					      edge from parent
				node[right] {$\theta_1$}
									}
        child  [sibling distance=40mm]{ node {$\mathtt{nat(y1), list(z)}$}
					 child  [sibling distance=40mm] { node {$\mathtt{list(z)}$}
            child  [sibling distance=20mm] { node {$\Box$}}
            child  [sibling distance=20mm] { node{$\mathtt{list(z1)}$}
					child  [sibling distance=10mm]{node{$\Box$}}
					child  [sibling distance=10mm] {node{$\vdots$}}}
								}			  
				child  [sibling distance=40mm]{ node {$\mathtt{nat(y2), list(z)}$}
					child  [sibling distance=30mm]{node{$\mathtt{list(z)}$}
					  child  [sibling distance=10mm]{node{$\Box$}}
						child  [sibling distance=10mm] {node{$\vdots$}}
						}
					child  [sibling distance=30mm] {node{$\mathtt{nat(y3), list(z)}$}
					child  [sibling distance=20mm] {node{$\vdots$}}}}
					}
					edge from parent
				node[right] {$\theta_0$}
					}; 
\end{tikzpicture}

\end{center}

\caption{\footnotesize{A proof tree and an $SLD$-tree for \texttt{ListNat} with the goal $list(x)$. 
A possible computed answer is given by the composition of $\theta_0=x/cons(y,z)$, $\theta_1= y/0$, $\theta_2= z/nil$;
Another computed answer is $\theta_4 = x/nil$.}}
\label{fig:SLD1}
\end{figure}

 The and-trees, or-trees and their combination have been used in parallel implementations of logic programming, \cite{GuptaC94,PontelliG95,GPACH12}.
The main idea was that branches in the $SLD$-trees and proof-trees can be exploited in parallel.  
For certain cases of logic programs, such as ground logic programs or
some fragments of DATALOG programs, one can do refutations for all the
atoms in the goal in parallel~\cite{Kanellakis88,UllmanG88}. But in
general, $SLD$-resolution is P-complete, and hence inherently
sequential~\cite{DKM84}.

The next definition formalises the notion of and-or parallel trees~\cite{GuptaC94,GPACH12}, but we restrict it to the ground cases, where such derivations are sound.


\begin{definition}\label{df:andortree}
  Let $P$ be a ground logic program and let $\ \gets A$ be an atomic goal
  (possibly with variables). The \emph{and-or parallel derivation
    tree} for $A$ is the possibly infinite tree $T$ satisfying the
  following properties.
\begin{itemize}
\item $A$ is the root of $T$.
\item Each node in $T$ is either an and-node or an or-node.
\item Each or-node is given by $\bullet$.
\item Each and-node is an atom.
\item For every node $A'$ occurring in $T$, if $A'$ is unifiable with only one clause $B \gets B_1, \ldots , B_n$ in $P$ with mgu $\theta$,
then $A'$ has $n$ children given by and-nodes $B_1\theta, \ldots B_n\theta$.
\item For every node $A'$ occurring in $T$, if $A'$ is unifiable with exactly $m>1$ distinct clauses $C_1, \ldots , C_m$ in $P$ via mgu's $\theta_1,\ldots , \theta_m$, 
then $A'$ has exactly $m$ children given by or-nodes, such that, for every $i \in \{1, \ldots ,m \}$, if $C_i  = B^i \gets B^i_1, \ldots ,B^i_n$, then the $i$th or-node has $n$ children given by and-nodes 
$B^i_1\theta_i, \ldots ,B^i_n\theta_i$.
\end{itemize}
\end{definition} 

Examples of and-or trees are given in Figures  \ref{pic:tree0} and \ref{pic:and-or}. 
In Section \ref{sec:parallel}, we return to the questions of parallelism for CoALP.

\section{Coalgebraic Semantics}\label{sec:OS}

In this section, we develop the coalgebraic semantics of logic programming, starting from the coalgebraic calculus of infinite trees, through to the observational semantics of SLD-derivations.

\subsection{A Coalgebraic Calculus of Infinite Trees}\label{sec:trees}

For the purposes of this paper, a {\em tree} $T$ consists of a set
$T_n$ for each natural number $n$, together with a function
$\delta_n:T_{n+1}\longrightarrow T_n$, yielding

\[
\ldots\; T_{n+1} \longrightarrow T_n \longrightarrow\; \ldots\; 
\longrightarrow T_1 \longrightarrow T_0 = 1
\]

An element of $T_n$ is called a {\em node} of $T$ at height $n$. The
unique element of $T_0$ is the {\em root} of the tree; for any $x\in
T_{n+1}$, $\delta_n(x)$ is called the {\em parent} of $x$, and $x$ is
called a {\em child} of $\delta_n(x)$. Observe that trees may have
infinite height, but if all $T_n$'s are finite, the tree is {\em
  finitely branching}.

An $L$-{\em labelled tree} is a tree $T$ together with a function
$l:\bigsqcup_{n \in \mathbb{N}}T_n \longrightarrow L$. The definitions of
$SLD$-tree and proof tree, Definitions~\ref{df:SLDtree}
and~\ref{df:prooftree} respectively, are of finitely branching
labelled trees. Both satisfy a further property: for any node $x$, the
children of $x$, i.e., the elements of $\delta^{-1}(x)$, have distinct
labels. This reflects the definition of a logic program,
following~\cite{Llo88}, as a set of clauses rather than as a list, and
the distinctness of atoms in the body of a clause. We accordingly say
an $L$-labelled tree is {\em locally injective} if for any node $x$,
the children of $x$ have distinct labels. Given a set $L$ of labels,
we denote the set of finitely branching locally injective $L$-labelled
trees by $Tree_L$.

We briefly recall fundamental constructs of
coalgebra, see also~\cite{JR97}.  

\begin{definition}\label{df:coalg}
For any endofunctor $H:C\longrightarrow C$, an
$H$-{\em coalgebra} consists of an object $X$ of $C$ together with a
map $x:X\longrightarrow HX$.  A {\em map} of $H$-coalgebras from
$(X,x)$ to $(Y,y)$ is a map $f:X\longrightarrow Y$ in $C$ such that
the diagram

\begin{center}
\begin{tikzpicture}[scale=2]
\draw (1,1) node {$X$}; 
\draw[->] (1.2,1) -- (1.8,1);
\draw (1,0) node {$HX$};
\draw[->] (1,0.9) -- (1,0.1);
\draw (2,1) node {$Y$};
\draw[->] (1.2,0) -- (1.8,0);
\draw (2,0) node {$HY$};
\draw[->] (2,0.9) -- (2,0.1);
\end{tikzpicture}
\end{center}
commutes. 
\end{definition}
$H$-coalgebras and maps of $H$-coalgebras form a 
category $H$-$coalg$, with composition determined by that in $C$,
together with a forgetful functor $U:H$-$coalg\longrightarrow C$, taking
an $H$-coalgebra $(X,x)$ to $X$. 

\begin{example}\label{ex:Pf}
Let $P_f$ denote the endofunctor on $Set$ that sends a set $X$ to the
set of its finite subsets, and sends a function $h:X\longrightarrow
Y$ to the function $P_f(h):P_f(X)\longrightarrow P_f(Y)$ sending a
subset $A$ of $X$ to its image $f(A)$ in $Y$. A $P_f$-coalgebra
$(X,x)$ is a finitely branching transition system, one of
the leading examples of coalgebra~\cite{JR97}.
\end{example} 

For any set $L$, the set $Tree_L$ of finitely branching locally
injective $L$-labelled trees possesses a canonical $P_f$-coalgebra
structure on it, sending $(T,l)$ to the set of $L$-labelled trees
determined by the children of the root of $T$. With mild overloading
of notation, we denote this $P_f$-coalgebra by $Tree_L$.

\begin{theorem}\label{thm:tree} The functor
  $U:P_f$-$Coalg\longrightarrow Set$ has a right adjoint sending any
  set $L$ to $Tree_L$.
\end{theorem}

\proof We have already seen that for any set $L$, the set $Tree_L$
possesses a canonical $P_f$-coalgebra structure given by sending an
$L$-labelled tree $(T,l)$ to the set of labelled trees determined by
the children of the root of $T$.

For the universal property, suppose we are given a $P_f$-coalgebra
$(X,x)$ and a function $h:X\longrightarrow L$. Put $h_0 =
h:X\longrightarrow L$.  For any $a\in X$, $x(a)$ is a finite subset of
$X$.  So $P_f(h_0)(x(a))$ is a finite subset of $L$.  Send $a$ to the
tree generated as follows: the root is labelled by $h_0(a)$; it has
$P_f(h_0)(x(a))$ children, each labelled by the corresponding element
of $P_f(h_0)(x(a))$; replace $h_0:X\longrightarrow L$ by $h_1 =
P_f(h_0)(x(-)):X\longrightarrow P_f(L)$, and continue inductively.

The unicity of this as a map of coalgebras is determined by its
construction together with the local injectivity condition; its
well-definedness follows from the finiteness of any element of
$P_f(X)$. $\Box$

We adapt this analysis to give a semantic account of the way in
which a logic program generates a tree of computations.

Given a set $L$ of labels, an $L$-{\em labelled} $\& \vee$-{\em tree}
is a finitely branching tree $T$ together with a function
$l:\bigsqcup_{n \in \mathbb{N}}T_{2n} \longrightarrow L$. In an
$L$-labelled $\&\vee$-tree, the nodes of even height are called
$\&$-nodes, and the nodes of odd height are called $\vee$-nodes. So
the $\&$-nodes, such as the root, are labelled, while the $\vee$-nodes
are not.

The and-or parallel derivation trees of Definition~\ref{df:andortree}
are labelled $\&\vee$-trees satisfying an additional property that
reflects logic programs consisting of sets rather than lists of
clauses and the distinctness of atoms in the body of a clause.  We
express the condition semantically as follows: an $L$-labelled
$\&\vee$-tree is {\em locally injective} if the children of any
$\vee$-node have distinct labels, and if, for any two distinct children
of an $\&$-node, the sets of labels of their children are distinct
(but may have non-trivial intersection), i.e., for any $x$, for any
$y,z\in \delta^{-1}(x)$, one has $l(\delta^{-1}(y)) \neq
l(\delta^{-1}(z))$. Given a set $L$ of labels, we denote the set of
locally injective $L$-labelled $\&\vee$-trees by $\&\vee$-$Tree_L$.

For any set $L$, the set $\&\vee$-$Tree_L$ has a canonical
$P_fP_f$-coalgebra structure on it, sending $(T,l)$ to the set of sets
of labelled $\&\vee$-trees given by the set of sets of $L$-labelled
$\&\vee$-trees determined by the children of each child of the root of
$T$. Again, we overload notation, using $\&\vee$-$Tree_L$ to denote
this coalgebra.

\begin{theorem}\label{thm:&vtree}
  The functor $U:P_fP_f$-$Coalg\longrightarrow Set$ has a right
  adjoint sending any set $L$ to $\&\vee$-$Tree_L$.
\end{theorem}

\proof A proof is given by a routine adaption of the proof of
Theorem~\ref{thm:tree}. $\Box$

There are assorted variants of Theorem~\ref{thm:&vtree}. We shall
need one for $L$-labelled $\&\vee_c$-trees, an $L$-labelled
$\&\vee_c$-tree being the generalisation of $L$-labelled $\&\vee$-tree
given by allowing countable branching at even heights, i.e., allowing
the root to have countably many children, but each child of the root
to have only finitely many children, etcetera.  Letting $P_c$ denote
the functor sending a set $X$ to the set of its countable subsets,
we have the following result.

\begin{theorem}\label{thm:&vctree}
  The functor $U:P_cP_f$-$Coalg\longrightarrow Set$ has a right
  adjoint sending any set $L$ to $\&\vee_c$-$Tree_L$.
\end{theorem}

\subsection{Coalgebraic Semantics for Ground Programs}\label{sec:derivation0}

Using our coalgebraic calculus of trees, we now make precise, in the
ground case, the relationship between logic programs and Gupta et al's
and-or parallel derivation trees of Definition~\ref{df:andortree}.

In general, if $U:H$-$coalg\longrightarrow C$ has a right adjoint $G$,
the composite functor $UG:C\longrightarrow C$ possesses the canonical
structure of a {\em comonad} $C(H)$, called the {\em cofree} comonad
on $H$. A {\em coalgebra} for a comonad is subtly different to a coalgebra for
an endofunctor as the former must satisfy two axioms, see also~\cite{BarWel90,LS86}. We denote
the category of $C(H)$-coalgebras by $C(H)$-$Coalg$. 

\begin{theorem}~\cite{JR97}\label{comonad} For any endofunctor
  $H:C\longrightarrow C$ for which the forgetful functor
  $U:H$-$coalg\longrightarrow C$ has a right adjoint, the category
  $H$-$coalg$ is canonically isomorphic to the category
  $C(H)$-$Coalg$. The isomorphism commutes with the forgetful functors
  to $C$.
\end{theorem}

Theorem~\ref{comonad} implies that for any $H$-coalgebra $(X,x)$,
there is a unique $C(H)$-coalgebra structure corresponding to it on
the set $X$.

Recall from the Introduction that, in the ground case, a logic program
can be identified with a coalgebra for the endofunctor $P_fP_f$ on
$Set$. By Theorem~\ref{thm:&vtree}, the forgetful functor
$U:P_fP_f$-$coalg\longrightarrow Set$ has a right adjoint taking a set
$L$ to the coalgebra $\&\vee$-$Tree_L$.  Thus the cofree comonad
$C(P_fP_f)$ on $P_fP_f$ sends the set $L$ to the set
$\&\vee$-$Tree_L$.

So Theorem~\ref{comonad} tells us that every ground logic program $P$
seen as a $P_fP_f$-coalgebra induces a canonical $C(P_fP_c)$-coalgebra
structure on the set $At$ of atoms underlying $P$, i.e., a function
from $At$ to $\&\vee$-$Tree_{At}$. 
\begin{theorem}\label{constr:Gcoalg}
  Given a $P_fP_f$-coalgebra $p:At\longrightarrow P_fP_f(At)$, the corresponding
  $C(P_fP_f)$-coalgebra has underlying set $At$ and action ${\bar
    p}:At\longrightarrow \&\vee$-$Tree_{At}$ as follows:

  For $A\in At$, the root of the tree ${\bar p}(A)$ is labelled by
  $A$.  If $p(A)\in P_fP_f(At)$ consists of $n$ subsets of $P_f(At)$,
  then the root of ${\bar p}(A)$ has $n$ children. The number and
  labels of each child of each of those $n$ children are determined by
  the number and choice of elements of $At$ in the corresponding
  subset of $P_f(A)$. Continue inductively.
\end{theorem}

\proof In general, for any endofunctor $H$ for which the forgetful functor
$U:H$-$coalg\longrightarrow C$ has a right adjoint $G$, the
$C(H)$-coalgebra induced by an $H$-coalgebra $(X,x)$ is given as follows:
$U(X,x) = X$, so the identity map $id:X\longrightarrow X$ can be
written as $id:U(X,x)\longrightarrow X$. By the definition of adjoint,
it corresponds to a map of the form
$\epsilon_{(X,x)}:(X,x)\longrightarrow GX$. Applying $U$ to
$\epsilon_{(X,x)}$ gives the requisite coalgebra map
$U\epsilon_{(X,x)}:X\longrightarrow C(H)X$.

Applying this to $H = P_fP_f$, this 
$C(P_fP_f)$-coalgebra structure is determined by the construction in
the proof of Theorem~\ref{thm:&vtree}, which is rewritten as the
assertion of this theorem. $\Box$

Comparing Theorem~\ref{constr:Gcoalg} with
Definition~\ref{df:andortree}, subject to minor reorganisation, given
a logic program $P$ seen as a $P_fP_f$-coalgebra, the corresponding
$C(P_fP_f)$-coalgebra structure on $At$ sends an atom $A$ to Gupta et
al's and-or parallel derivation tree, characterising their
construction in the ground case.

\begin{example}\label{ex:lptree}
Consider the ground logic program
\begin{eqnarray*}
\texttt{q(b,a)} & \gets & \\ 
\texttt{s(a,b)} & \gets & \\
\texttt{p(a)} & \gets & \texttt{q(b,a)},\texttt{s(a,b)} \\
\texttt{q(b,a)}  & \gets &  \texttt{s(a,b)}\\
\end{eqnarray*}

The program has three atoms, namely \texttt{q(b,a)}, \texttt{s(a,b)}
and \texttt{p(a)}.\\
So $At = \{ \texttt{q(b,a)},\texttt{s(a,b)},\texttt{p(a)}\}$. 
The program can be identified with the $P_fP_f$-coalgebra structure on $At$ given by\\ 
$p(\texttt{q(b,a)}) = \{ \{\empty\}, \{\texttt{s(a,b)}\}\}$,
where $\{\empty \}$ is the empty set.\\
$p(\texttt{s(a,b)}) = \{ \{\empty\}\}$, i.e., the one element set consisting of the empty set.\\
$p(\texttt{p(a)}) = \{\{\texttt{q(b,a),s(a,b)}\}\}$.

The corresponding $C(P_fP_f)$-coalgebra sends \texttt{p(a)} to the
parallel refutation of \texttt{p(a)} depicted on the left side of
Figure~\ref{pic:tree0}. Note that the nodes of the tree alternate
between those labelled by atoms and those labelled by $\bullet$. The
set of children of each $\bullet$ represents a goal, made up of the
conjunction of the atoms in the labels. An atom with multiple children
is the head of multiple clauses in the program: its children represent
these clauses. We use the traditional notation $\Box$ to denote
$\{\empty\}$.

Where an atom has a single $\bullet$-child, we can elide that node
without losing any information; the result of applying this
transformation to our example is shown on the right side of
Figure~\ref{pic:tree0}. The resulting tree is precisely the and-or
parallel derivation tree for the atomic goal
$\gets\mathtt{p(a)}$. 
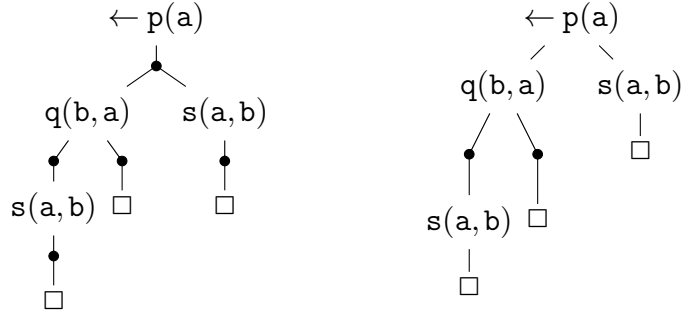
\begin{figure}[!h]
\begin{center}
  \begin{tikzpicture}[scale=0.9,baseline=(current bounding box.north),grow=down,level distance=7mm, sibling distance
    = 10mm]
\node {$\gets\mathtt{p(a)}$}
child { [fill] circle (2pt)
  child { node {$\mathtt{q(b,a)}$}
     child {[fill] circle (2pt)
         child {node {$\mathtt{s(a,b)}$}
            child {[fill] circle (2pt)
               child { node {$\Box$} }
            }
         }
     } 
     child {[fill] circle (2pt)
       child { node {$\Box$}} 
     }
  } 
  child[fill=none] {edge from parent[draw=none]} 
  child {node {$\mathtt{s(a,b)}$}
    child {[fill] circle (2pt)
       child { node {$\Box$} }
    }
  }
}; 
  \end{tikzpicture}
\quad\quad\quad\quad
  \begin{tikzpicture}[scale=0.9, baseline=(current bounding box.north),grow=down,level distance=10mm, sibling distance = 10mm]
\node {$\gets\mathtt{p(a)}$}
  child { node {$\mathtt{q(b,a)}$}
     child {[fill] circle (2pt)
         child {node {$\mathtt{s(a,b)}$}
                    child { node {$\Box$} }
              }
     } 
     child {[fill] circle (2pt)
       child { node {$\Box$}} 
     }
  } 
  child[fill=none] {edge from parent[draw=none]} 
  child {node {$\mathtt{s(a,b)}$}
       child { node {$\Box$} }
  }; 
  \end{tikzpicture}
\end{center}
\caption{\footnotesize{The action of $\overline{p}: \At \longrightarrow
  C(P_fP_f)(\At)$ on $\texttt{p(a)}$ and the corresponding and-or
  parallel derivation tree.}}
\label{pic:tree0} 
\end{figure}
\end{example}
 
\subsection{Coalgebraic Semantics for Arbitrary Programs}\label{sec:coalg} 

Extending from ground logic programs to first-order programs is not
routine. Following normal category theoretic practice, we model the
first-order language underlying a logic program by a Lawvere
theory~\cite{AmatoLM09,BonchiM09,BruniMR01}.

\begin{definition}\label{Law}
  Given a signature $\Sigma$ of function symbols, the {\em Lawvere
    theory} $\mathcal{L}_{\Sigma}$ generated by $\Sigma$ is the
  following category: $\texttt{ob}(\mathcal{L}_{\Sigma})$ is the set
  of natural numbers.  For each natural number $n$, let $x_1,\ldots
  ,x_n$ be a specified list of distinct variables. Define
  $\texttt{ob}(\mathcal{L}_{\Sigma})(n,m)$ to be the set of $m$-tuples
  $(t_1,\ldots ,t_m)$ of terms generated by the function symbols in
  $\Sigma$ and variables $x_1,\ldots ,x_n$. Define composition in
  $\mathcal{L}_{\Sigma}$ by substitution.
%
\end{definition}
%
%

One can describe $\mathcal{L}_{\Sigma}$ without the need for a
specified list of variables for each $n$: in a term $t$, a variable
context is always implicit, i.e., $x_1,\ldots ,x_m\vdash t$, and the
variable context may be considered as a binder.

For each signature $\Sigma$, we extend the set
$At$ of atoms for a ground logic program to the functor $At:\ls^{op}
\rightarrow Set$ that sends a natural number $n$ to the set of all
atomic formulae generated by $\Sigma$, variables among a fixed set
$x_1,\ldots ,x_n$, and
the predicate symbols appearing in the logic program. A map $f:n
\rightarrow m$ in $\ls$ is sent to the function $At(f):At(m)
\rightarrow At(n)$ that sends an atomic formula $A(x_1, \ldots,x_m)$
to $A(f_1(x_1, \ldots ,x_n)/x_1, \ldots ,f_m(x_1, \ldots ,x_n)/x_m)$,
i.e., $At(f)$ is defined by substitution.

Given a logic program $P$ with function symbols in $\Sigma$, we would
like to model $P$ by the putative $[\ls^{op},P_fP_f]$-coalgebra
$p:At\longrightarrow P_fP_fAt$ on the category $[\ls^{op},Set]$ whose
$n$-component takes an atomic formula $A(x_1,\ldots ,x_n)$ with at
most $n$ variables, considers all substitutions of clauses in $P$
whose head agrees with $A(x_1,\ldots ,x_n)$, and gives the set of sets
of atomic formulae in antecedents. Unfortunately, it does not work.

Consider the logic program \texttt{ListNat} of Example~\ref{ex:listnat}. There
is a map in $\mathcal{L}_{\Sigma}$ of the form $0\rightarrow 1$ that
models the nullary function symbol $0$. Naturality of the map
$p:At\longrightarrow P_fP_fAt$ in $[\mathcal{L}_{\Sigma}^{op},Set]$
yields commutativity of the diagram

\begin{center}
\begin{tikzpicture}[scale=2]
\draw (1,1) node {$At(1)$}; 
\draw[->] (1.3,1) -- (1.5,1);
\draw (1,0) node {$At(0)$};
\draw[->] (1,0.8) -- (1,0.2);
\draw (2,1) node {$P_fP_fAt(1)$};
\draw[->] (1.3,0) -- (1.5,0);
\draw (2,0) node {$P_fP_fAt(0)$};
\draw[->] (2,0.8) -- (2,0.2);
\end{tikzpicture}
\end{center}
There being no clause of the form $\mathtt{nat(x)}\gets \, $ in
\texttt{ListNat}, commutativity implies
that there cannot be a clause in \texttt{ListNat} of the form
$\mathtt{nat(0)}\gets \, $ either, but in fact there is one. 

We resolve this by relaxing the naturality condition on $p$ to a
subset condition, yielding lax naturality. To define it, we extend
$At:\ls^{op}\rightarrow Set$ to have codomain $Poset$, which we do by
composing $At$ with the inclusion of $Set$ into $Poset$. Mildly
overloading notation, we denote the composite by
$At:\ls^{op}\rightarrow Poset$.

$Poset$ canonically possesses the structure of a locally ordered
category, i.e., there is a canonical partial order on each homset
$Poset(P,Q)$ and it is respected by composition. It is given
pointwise: $f\leq g$ if and only if for all $x\in P$, one has
$f(x)\leq g(x)$ in $Q$.  The category $\mathcal{L}_{\Sigma}$ also has
a canonical locally ordered structure given by the discrete
structure, i.e., $f\leq g$ if and only if $f = g$. Any functor
from $\mathcal{L}_{\Sigma}^{op}$ to $Poset$ is trivially locally ordered,
i.e., preserves the partial orders.


\begin{definition}
  Given locally ordered functors $H,K:D \longrightarrow C$, a {\em lax
    natural transformation} from $H$ to $K$ is the assignment to each
  object $d$ of $D$, of a map $\alpha_d: Hd \longrightarrow Kd$ such
  that for each map $f:d \longrightarrow d'$ in $D$, one has
  $(Kf)(\alpha_d) \leq (\alpha_{d'})(Hf)$.
\end{definition}
%
%
Locally ordered functors and lax natural transformations, with
pointwise composition and pointwise ordering, form a locally ordered
category we denote by $Lax(D,C)$.  

A final problem arises in regard to the finiteness of the outer
occurrence of $P_f$ in $P_fP_f$. The problem is that substitution can
generate infinitely many instances of clauses with the same head. For
instance, if one extends \texttt{ListNat} with a clause of the form $A
\gets \mathtt{nat(x)}$ with no occurrences of $x$ in $A$, substitution
yields the clause $A \gets \mathtt{nat(s^n(0))}$ for every natural
number $n$, giving rise to a countable set of clauses with head
$A$. Graph connectivity, \texttt{GC}, gives another example, see
Example~\ref{ex:lp}.

We address this issue by replacing $P_fP_f$ by $P_cP_f$, where $P_c$
is the countable powerset functor, extending $P_cP_f$ from $Set$ to a
locally ordered endofunctor on $Poset$, upon which composition yields
the locally ordered endofunctor we seek on $Lax(\ls^{op},Poset)$.

\begin{definition}\label{df:po} Define $P_f:Poset\longrightarrow Poset$ 
  by letting $P_f(P)$ be the partial order given by the set of finite
  subsets of $P$, with $A\leq B$ if for all $a \in A$, there exists $b
  \in B$ for which $a\leq b$ in $P$, with behaviour on maps given by
  image. Define $P_c$ similarly but with countability replacing
  finiteness.
\end{definition}

A cofree comonad $C(P_cP_f)$ exists on $P_cP_f$ and, by
Theorem~\ref{thm:&vctree}, we can describe it: $C(P_cP_f)(P) =
\&\vee_c$-$Tree_P$, with partial order structure generated by
Definition~\ref{df:po}. In order to extend the correspondence between
$P_cP_f$-coalgebras $p:At\longrightarrow P_cP_fAt$ and
$C(P_cP_f)$-coalgebras ${\bar p}:At\longrightarrow C(P_cP_f)At$ from
$Poset$ to $Lax(\ls^{op},Poset)$, we need to do some abstract category
theory.

Let $H$ be an arbitrary locally ordered endofunctor on
an arbitrary locally ordered category $C$.  Denote by
$H\mbox{-}coalg_{oplax}$ the locally ordered category whose objects
are $H$-coalgebras and whose maps are oplax maps of $H$-coalgebras,
meaning that, in the square

\begin{center}
\begin{tikzpicture}[scale=2]
\draw (1,1) node {$X$}; 
\draw[->] (1.2,1) -- (1.8,1);
\draw (1,0) node {$HX$};
\draw[->] (1,0.9) -- (1,0.1);
\draw (1.5,0.5) node {$\leq$};
\draw (2,1) node {$Y$};
\draw[->] (1.2,0) -- (1.8,0);
\draw (2,0) node {$HY$};
\draw[->] (2,0.9) -- (2,0.1);
\end{tikzpicture}
\end{center}
%
%
the composite via $HX$ is less than or equal to the composite via
$Y$. Since $C$ and $H$ are arbitrary, one can replace $C$ by
$Lax(D,C)$, for any category $D$; and replace $H$ by $Lax(D,H)$, yielding the locally ordered
category $Lax(D,H)\mbox{-}coalg_{oplax}$.

\begin{proposition}\label{prop:2}
  The locally ordered category $Lax(D,H)\mbox{-}coalg_{oplax}$ is
  canonically isomorphic to $Lax(D,H\mbox{-}coalg_{oplax})$.
\end{proposition}

\proof
Unwinding the definitions, to give a functor $J:D\longrightarrow
H$-$coalg_{oplax}$ is, by definition, to give, for each object $d$ of
$D$, a map in $C$ of the form $Jd:J_0d\longrightarrow HJ_0d$, and,
for each map $f:d\longrightarrow d'$ in $D$, a map in $C$ of
the form $J_0f:J_0d\longrightarrow J_0d'$, such that


$$
\xymatrix@R=1.0pc@C=1.5pc{
*{J_0d}\ar[rrrr]^{J_0f}\ar[dddd]_{Jd}&&&& *{J_0d'}\ar[dddd]^{Jd} \\
&&&&\\
&&\leq&&\\
&&&&\\
*{HJ_0d}\ar[rrrr]_{HJ_0f} &&&& *{HJ_0d'}}
$$
subject to locally ordered functoriality equations.

These data and axioms can be re-expressed as a locally ordered functor
$J_0:D\longrightarrow C$ together with a lax natural transformation
$j:J_0\longrightarrow HJ_0$, the condition for lax naturality of $j$ in
regard to the map $f$ in $D$ being identical to the condition that
$J_0f$ be an oplax map of coalgebras from $Jd$ to $Jd'$. 

This yields a canonical bijection between the sets of
objects of \mbox{$Lax(D,H$-$coalg_{oplax})$} and
$Lax(D,H)$-$coalg_{oplax}$, that bijection canonically extending to
a canonical isomorphism of locally ordered categories.
$\Box$

\begin{proposition}\label{Prop:3}
  Given a locally ordered comonad $G$ on a locally ordered category
  $C$, the data given by $Lax(D,G):Lax(D,C) \rightarrow Lax(D,C)$ and
  pointwise liftings of the structural natural transformations of $G$
  yield a locally ordered comonad we also denote by $Lax(D,G)$ on
  $Lax(D,C)$.
\end{proposition}

\proof This holds by tedious but routine checking of all the axioms in the
definition of locally ordered comonad. $\Box$\\
%
Given a locally ordered comonad $G$, denote by
$G\mbox{-}Coalg_{oplax}$ the locally ordered category whose objects
are $G$-coalgebras and whose maps are oplax maps of $G$-coalgebras.

\begin{proposition}\label{prop:4}
  Given a locally ordered comonad $G$, $Lax(D,G)\mbox{-}Coalg_{oplax}$
  is canonically isomorphic to $Lax(D,G\mbox{-}Coalg_{oplax})$.
\end{proposition}

\proof A proof is given by routine extension of the proof of
Proposition~\ref{prop:2}. $\Box$

\begin{theorem}~\cite{Kelly74}\label{thm:5} Given a locally ordered
  endofunctor $H$ on a locally ordered category with finite colimits
  $C$, if $C(H)$ is the cofree comonad on $H$, then
  $H\mbox{-}coalg_{oplax}$ is canonically isomorphic to
  $C(H)\mbox{-}Coalg_{oplax}$.
\end{theorem}

Combining Proposition~\ref{prop:2}, Proposition~\ref{prop:4} and
Theorem~\ref{thm:5}, we can conclude the following:

\begin{theorem}\label{thm:good}
  Given a locally ordered endofunctor $H$ on a locally ordered
  category with finite colimits $C$, if $C(H)$ is the cofree comonad
  on $H$, then there is a canonical isomorphism
\[
Lax(D,H)\mbox{-}Coalg_{oplax} \simeq Lax(D,C(H))\mbox{-}Coalg_{oplax}
\]
\end{theorem}

\begin{corollary}\label{main}
For any locally ordered endofunctor $H$ on $Poset$, if $C(H)$ is
the cofree comonad on $H$, then there is a canonical isomorphism
\[
Lax(\ls^{op},H)\mbox{-}Coalg_{oplax} \simeq
Lax(\ls^{op},C(H))\mbox{-}Coalg_{oplax}
\]
\end{corollary}

Putting $H = P_cP_f$, Corollary~\ref{main} gives us the abstract
result we need. The lax natural transformation $p:At\longrightarrow
P_cP_fAt$ generated by a logic program $P$, evaluated at a natural
number $n$, sends an atomic formula $A(x_1,\ldots ,x_n)$ to the set of
sets of antecedents in substitution instances of clauses in $P$ for
which the head of the substituted instance agrees with $A(x_1,\ldots
,x_n)$.  That in turn yields a lax natural transformation ${\bar p}:At
\longrightarrow C(P_cP_f)At$, which, evaluated at $n$, is the function
from the set $At(n)$ to the set $\&\vee_c$-$Tree_{At(n)}$ determined
by the construction of Theorem~\ref{constr:Gcoalg} if one treats the
variables $x_1,\ldots ,x_n$ as constants. See also \cite{BonchiZ13} for a Laxness-free semantics for CoALP.

\begin{example}\label{ex:treefo}
  Consider \texttt{ListNat} as in Example~\ref{ex:listnat}. Suppose we
  start with $A(x,y)\in At(2)$ given by the atomic formula
  $\mathtt{list(cons(x,cons(y,x)))}$. Then ${\bar p}(A(x,y))$ is the
  element of $C(P_cP_f)At(2) = \&\vee_c$-$Tree_{At(2)}$ expressible by
    the tree on the left hand side of Figure~\ref{pic:tree}.

    This tree agrees with the start of the and-or parallel
    derivation tree for $\mathtt{list(cons(x,cons(y,x)))}$. It has
    leaves \texttt{nat(x)}, \texttt{nat(y)} and \texttt{list(x)},
    whereas the and-or parallel derivation tree follows those nodes,
    using substitutions determined by mgu's that might not be
    consistent with each other, e.g., there is no consistent
    substitution for \texttt{x}.

  Lax naturality means a substitution potentially yields two different
  trees: one given by substitution into the tree, then pruning to
  remove redundant branches, the other given by substitution into the
  root, then applying ${\bar p}$.

  For example, we can substitute $s(z)$ for both $x$ and $y$ in
  $\mathtt{list(cons(x,cons(y,x)))}$. This substitution is given by
  applying $At$ to the map $(s,s):1 \longrightarrow 2$ in $\ls$. So
  $At((s,s))(A(x,y))$ is an element of $At(1)$. Its image under ${\bar
    p} (1):At(1)\longrightarrow C(P_cP_f)At(1)$ is the element of
  $C(P_cP_f)At(1) = \&\vee_c$-$Tree_{At(1)}$ given by the tree in
  the middle of Figure~\ref{pic:tree}.
%
%
\begin{figure}
\begin{center}
  \begin{tikzpicture}[scale=0.35,baseline=(current bounding box.north),grow=down,level distance=16mm,sibling distance=50mm,font=\footnotesize]
  \node {$\mathtt{list(c(x,c(y,x)))}$}
   child {[fill] circle (4pt)
     child { node {$\mathtt{nat(x)}$}}
       child { node {$\mathtt{list(c(y,x))}$}
        child {[fill] circle (4pt)
       child { node{$\mathtt{nat(y)}$}}
         child { node{$\mathtt{list(x)}$}}}
       }}; 
  \end{tikzpicture}
$\rightarrow$
\begin{tikzpicture}[scale=0.35,baseline=(current bounding box.north),grow=down,level distance=16mm,sibling distance=60mm,font=\footnotesize ]
  \node {$\mathtt{list(c(s(z),c(s(z),s(z))))}$}
   child {[fill] circle (4pt)
     child [sibling distance=70mm]{ node {$\mathtt{nat(s(z))}$}
     child {[fill] circle (4pt)
     child { node {$\mathtt{nat(z)}$}
     }}}
       child { node {$\mathtt{list(c(s(z),s(z)))}$}
        child {[fill] circle (4pt)
          child [sibling distance=50mm]{ node{$\mathtt{nat(s(z))}$}
            child {[fill] circle (4pt)
             child { node {$\mathtt{nat(z)}$}}}}
             child [sibling distance=50mm]{ node{$\mathtt{list(s(z))}$}}}
       }}; 
  \end{tikzpicture}
	$\rightarrow$
  \begin{tikzpicture}[scale=0.35,baseline=(current bounding box.north),grow=down,level distance=16mm,sibling distance=35mm,font=\footnotesize]
  \node {$\mathtt{list(c(s(0),c(s(0),s(0))))}$}
   child {[fill] circle (4pt)
     child [sibling distance=75mm]{ node {$\mathtt{nat(s(0))}$}
       child {[fill] circle (4pt)
         child { node {$\mathtt{nat(0)}$}
            child {[fill] circle (4pt)
               child { node {$\Box$}}}}}}
      child [sibling distance=55mm] { node {$\mathtt{list(c(s(0),s(0)))}$}
        child {[fill] circle (4pt)
          child [sibling distance=50mm]{ node{$\mathtt{nat(s(0))}$}
            child {[fill] circle (4pt)
             child { node {$\mathtt{nat(0)}$}
               child {[fill] circle (4pt)
                child { node {$\Box$}}}}}}
             child[sibling distance=50mm] { node{$\mathtt{list(s(0))}$}}}
       }}; 
  \end{tikzpicture}
\end{center}
\caption{\footnotesize{The left hand tree depicts
  ${ \bar p}(\mathtt{list(cons(x,cons(y,x)))})$, the middle tree 
depicts ${\bar p} At(s,s)(\mathtt{list(cons(x,cons(y,x)))})$, i.e., 
${ \bar p}(\mathtt{list(cons(s(z),cons(s(z),s(z))))})$, and the right tree
depicts ${ \bar p} At(0)At(s,s)(\mathtt{list(cons(x,cons(y,x)))})$;
\texttt{cons} is abbreviated by \texttt{c}.}}
\label{pic:tree} 
\end{figure}
The laxness of the naturality of ${\bar p}$ is indicated by the
increased length, in two places, of this tree. Before those two
places, the two trees have the same structure.

Now suppose we make the further substitution of $0$ for $z$. This
substitution is given by applying $At$ to the map $0:0\rightarrow 1$
in $\ls$.  In Figure~\ref{pic:tree}, we depict ${\bar p} (0)
At(0)At((s,s))(A(x,y))\in \&\vee_c$-$Tree_{At(1)}$ on the right.
Two of the leaves of the latter tree are labelled by $\Box$, but one
leaf, namely $\mathtt{list(s(0))}$ is not, so the tree does not yield
a proof. Again, observe the laxness.
\end{example}

This requires care. Consider the following example,
studied in~\cite{SS86}.

\begin{example}[GC] \label{ex:lp}
Let \texttt{GC}
(for graph connectivity) denote the logic program
\begin{eqnarray*}
\texttt{connected(x,x)} & \gets &\\
\texttt{connected(x,y)} & \gets & \texttt{edge(x,z)},
\texttt{connected(z,y)}.\\
\end{eqnarray*}
There may be additional function symbols, such as a unary $s$, and
additional clauses to give a database, such as $\texttt{edge(0,s(0))}\gets$
and $\texttt{edge(s(0),s(s(0)))}\gets\;$.
Note the presence of a variable $z$ in the body but not the head of
the clause 
\[
\texttt{connected(x,y)} \gets \texttt{edge(x,z)},
\texttt{connected(z,y)}
\]  
That allows derivations involving
infinitely many variables, thus not directly yielding a subtree of
${\bar p}(\texttt{connected(x,y)})\in \&\vee_c$-$Tree_{At(n)}$ for any $n$.
\end{example}

The subtle relationship between the finite and the infinite illustrated
by Example~\ref{ex:lp} is fundamental to the idea of coalgebraic
logic programming, which we develop in the latter sections of the paper. See also Figure \ref{pic:stream}.


\begin{definition}\label{df:derivsub}
  Let $P$ be a logic program, $G$ be an atomic goal, and $T$ be the
  $\&\vee_c$-tree  determined by $P$ and $G\in At(n)$.  A subtree
  $T'$ of $T$ is called a \emph{derivation subtree} of $T$ if it
  satisfies the following conditions:
\begin{itemize}
\item the root of $T'$ is the root of $T$ (up to variable renaming);
\item if an and-node belongs to $T'$, then at most one of its children belongs
  to $T'$.
\item if an or-node belongs to $T'$, then all its children belong to $T'$.
\end{itemize}
A finite derivation tree is {\em successful} if its leaves are all
or-nodes (equivalently, they are followed only by $\Box$ in the usual
pictures). 
\end{definition}

By Example~\ref{ex:lp}, derivations need not directly yield derivation
subtrees. Nevertheless, all subderivations of finite length of a
derivation do form derivation subtrees.

\begin{theorem}[Soundness and Completeness of $SLD$-refutations]\label{th:sc}
Let $P$ be a logic program, and $G$ be an atomic goal.
\begin{enumerate}
\item Soundness. If there is an $SLD$-refutation for $G$ in $P$ with
  computed answer $\theta$, then for some $n$ with $G\theta\in At(n)$,
  the $\&\vee_c$-tree for $G\theta$ contains a successful derivation
  subtree.
\item Completeness. If the $\&\vee_c$-tree for $G\theta\in At(n)$
  contains a successful derivation subtree, then there exists an
  $SLD$-refutation for $G$ in $P$, with computed answer $\lambda$ for
  which $ \lambda \sigma = \theta $ for some $\sigma$.
\end{enumerate}
\end{theorem}

\begin{proof}
The finiteness of refutations makes this a routine adaptation of
the soundness and completeness of the collectivity of $SLD$-trees
for $SLD$-refutation.
\end{proof}

\subsection{Coalgebraic Semantics and the Theory of Observables}\label{sec:TO}

Our coalgebraic analysis relates closely to the \emph{Theory of
  Observables} for logic programming developed in
\cite{CominiLM01}. In that theory, the traditional characterisation of
logic programs in terms of input/output behaviour and successful
derivations is not sufficient for the purposes of program analysis and
optimisation. One requires more complete information about
$SLD$-derivations, specifically the sequences of goals and most
general unifiers used. Infinite derivations can be meaningful.  The
following observables are critical to the theory
\cite{GLM95,CominiLM01}.

\begin{definition}\label{df:observe}
\begin{enumerate}
\item A \emph{call pattern} is a sequence of atoms selected in an
  $SLD$-derivation; a \emph{correct call pattern} is a sequence of
  atoms selected in an $SLD$-refutation.
\item A \emph{partial answer} is a substitution associated with a
  resolvent in an $SLD$-derivation; a \emph{correct partial answer}
  is a substitution associated with a resolvent in an $SLD$-refutation.
\end{enumerate}
\end{definition}

As explained in \cite{GLM95,CominiLM01}, semantics of logic
programs aims to identify observationally equivalent logic programs and to
distinguish logic programs that are not observationally equivalent.
So the definitions of observation and semantics are interdependent. 
Observational equivalence was defined in~\cite{GLM95} as follows. 

\begin{definition}\label{df:observeq}
  Let $P_1$ and $P_2$ be logic programs with the same alphabet. Then
  $P_1$ is {\em observationally equivalent} to $P_2$, written $P_1
  \approx P_2$, if for any goal $G$, the following conditions
  hold:
\begin{enumerate}
\item $G$ has an $SLD$-refutation in $P_1$ if and only if $G$ has an
  $SLD$-refutation in $P_2$.
\item $G$ has the same set of computed answers in $P_1$ as in $P_2$.
\item $G$ has the same set of (correct) call patterns in $P_1$ as in $P_2$.
\item $G$ has the same set of (correct) partial answers in $P_1$ as in $P_2$.
\end{enumerate}

\end{definition}

\begin{theorem}[Correctness]\label{th:observeq}
  For logic programs $P_1$ and $P_2$, if the
  $Lax(\ls^{op},C(P_cP_f))$-coalgebra structure ${\bar p_1}$ generated
  by $P_1$ is isomorphic to the $Lax(\ls^{op},C(P_cP_f))$-coalgebra
  structure ${\bar p_2}$ generated by $P_2$ (denoted $\bar{p_1} \cong \bar{p_2}$), then $P_1 \approx P_2$.
\end{theorem}

The converse of Theorem \ref{th:observeq}, \emph{full abstraction},
does not hold, i.e., with the above definition of observational
equivalence, there are observationally equivalent programs that have
different $\&\vee_c$-$Tree$s.

\begin{example}
  Consider logic programs $P_1$ and $P_2$ with the same clauses except
  for one: $P_1$ contains $A \gets B_1,\texttt{false},B_2$; and $P_2$
  contains the clause $A \gets B_1,\texttt{false}$ instead.  The atoms
  in the clauses are such that $B_1$ has a refutation in
  $P_1$ and $P_2$, and \texttt{false} is an atom that has no
  refutation in the programs.  In this case, assuming a left-to-right
  sequential evaluation strategy, all derivations that involve the two
  clauses in $P_1$ and $P_2$ will always fail on \texttt{false}, and
  $P_1$ will be observationally equivalent to $P_2$, but they generate
different trees because of $B_2$.
\end{example}

We can recover full abstraction if we adapt Definitions
\ref{df:observe} and \ref{df:observeq} so that they do not rely upon
an algorithm to choose a selected atom but rather allow arbitrary
choices. This
is typical of coalgebra, yielding essentially an instance of
bisimulation~\cite{JR97}. In order to do that, we need to modify
Definitions~\ref{df:SLD} and~\ref{df:SLDderiv} to eliminate the
algorithm used in the definitions leading to $SLD$-derivations. 

\begin{definition}\label{df:coind-res}
  Let a goal $G$ be $\gets A_1,\ldots, A_k$ and a clause $C$ be
  $A\gets B_1, \ldots ,B_q$. Then $G'$ is \emph{non-deterministically derived}
  from $G$ and $C$ using mgu $\theta$ if the following conditions
  hold:
\item[$\bullet$] $\theta$ is an mgu of some
atom $A_m$  in the body of $G$ and $A$;
\item[$\bullet$] $G'$ is the goal $\gets(A_1, \ldots, A_{m-1},B_1, \ldots
,B_q, A_{m+1},
\ldots, A_k)\theta$.
\end{definition}

Definition~\ref{df:coind-res} differs from Definition~\ref{df:SLD}
in precisely one point: the former refers to ``some atom'' where the
latter refers to ``the selected atom'', with the selection being
determined by an algorithm. The distinction means that
Definition~\ref{df:coind-res} has nondeterminism built into the choice
of atom, which in turn implies the possibility of parallelism in
implementation. We will exploit that later. 
It further implies that a verbatim restatement of Definition~\ref{df:observeq}
but with ``$SLD$-derivation'' replaced by ``coinductive derivation'' also
implies the possibility of implementation based on parallelism.

\begin{definition}\label{df:coind-der}
  A \emph{non-deterministic derivation} of $P\cup \{G\}$ consists of a
  sequence of goals $G=G_0, G_1, \ldots$ called \emph{non-deterministic
    resolvents}, a sequence $C_1,C_2,\ldots$ of variants of program
  clauses of $P$, and a sequence $\theta_1,\theta_2,\ldots$ of mgu's
  such that each $G_{i+1}$ is derived from $G_i$ using
  $\theta_{i+1}$. A \emph{non-deterministic refutation} of $P \cup \{G\}$ is
  a finite non-deterministic derivation of $P \cup \{G\}$ such that its last
  goal is empty, denoted by $\Box$.  If $G_n = \Box$, we say that the
  refutation has length $n$. The composite $\theta_1\theta_2\ldots$
is called a {\em computed answer}.
\end{definition}

Figure  \ref{pic:stream-n}  exhibits a non-deterministic derivation for the goal $G = \texttt{stream(x)}$ 
and the program \texttt{Stream} from Example \ref{ex:stream}.
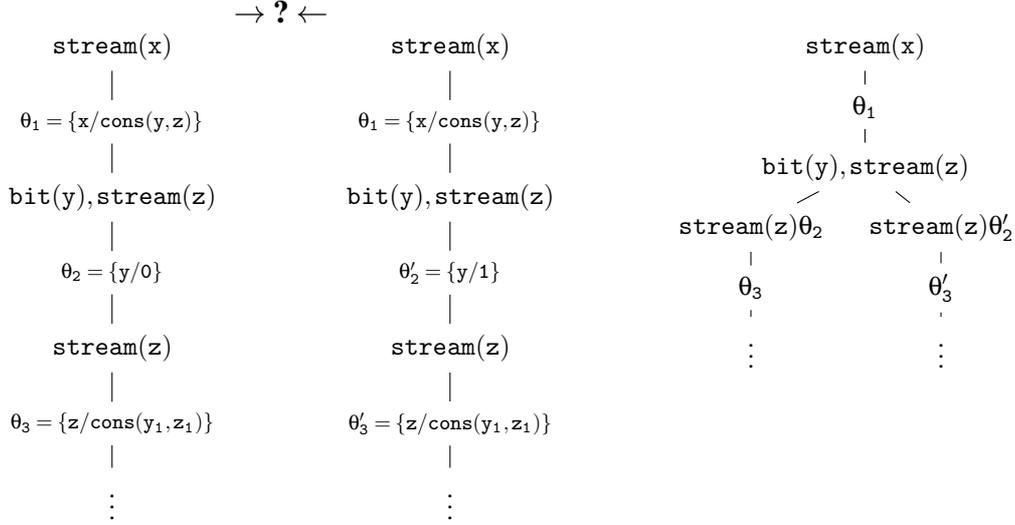
\begin{figure}
\begin{center}

\begin{tikzpicture}[scale=1,baseline=(current bounding box.north),grow=down,level distance=10mm,sibling distance=15mm,font=\footnotesize]
  \node {$\mathtt{stream(x)}$}
        child { node {\emph{\scriptsize{$\mathtt{\theta_1=\{x/cons(y,z)\}}$}}}
     child { node {$\mathtt{bit(y),stream(z)}$}
        child { node {\emph{\scriptsize{$\mathtt{\theta_2=\{y/0\}}$}}}
		   child { node {$\mathtt{stream(z)}$}
		                        child {node {\emph{\scriptsize{$\mathtt{\theta_3=\{z/cons(y_1,z_1)\}}$}}}
								child {node {$\vdots$}}}}}}}; 
  \end{tikzpicture}
  $\rightarrow \textbf{?} \leftarrow$
\begin{tikzpicture}[scale=1,baseline=(current bounding box.north),grow=down,level distance=10mm,sibling distance=15mm,font=\footnotesize]
  \node {$\mathtt{stream(x)}$}
        child { node {\emph{\scriptsize{$\mathtt{\theta_1=\{x/cons(y,z)\}}$}}}
     child { node {$\mathtt{bit(y),stream(z)}$}
        child { node {\emph{\scriptsize{$\mathtt{\theta_2\rq{}=\{y/1\}}$}}}
		   child { node {$\mathtt{stream(z)}$}
		                        child {node {\emph{\scriptsize{$\mathtt{\theta_3\rq{}=\{z/cons(y_1,z_1)\}}$}}}
								child {node {$\vdots$}}}}}}}; 
  \end{tikzpicture}
    \quad\quad\quad
\begin{tikzpicture}[baseline=(current bounding box.north),grow=down,level distance=8mm,sibling distance=20mm,font=\footnotesize]
    \node {$\mathtt{stream(x)}$}
     child { node {\emph{\footnotesize{$\mathtt{\theta_1}$}}}
      child [sibling distance=20mm]{ node {$\mathtt{bit(y),stream(z)}$}
         child  [sibling distance=30mm] { node {$\mathtt{stream(z)\theta_2}$}
              						child [sibling distance=20mm] { node{\emph{\footnotesize{$\mathtt{\theta_3}$}}}
							   child  [sibling distance=20mm]{node{$\vdots$}}}}
					            child  [sibling distance=20mm]{ node {$\mathtt{stream(z)\theta_2\rq{}}$}
					            		                        child {node {\emph{\footnotesize{$\mathtt{\theta_3\rq{}}$}}}
					child  [sibling distance=15mm] {node{$\vdots$}}}
					}				}}; 
\end{tikzpicture}
\end{center}
\caption{\footnotesize{\textbf{Left:} Two possible choices for non-deterministic derivation for the goal $G = \texttt{stream(x)}$ and 
the program \texttt{Stream}, with $\theta_1 = \{x/scons(y,z)\}$, $\theta_2 = \{y/0\}$, $\theta_2\rq{}=\{y/1\}$ and $\theta_3 = \theta_3\rq{}=\{y/scons(y_1,z_1)\}$. \textbf{Right:} the two non-deterministic derivations shown in the form of an SLD-tree.}}
\label{pic:stream-n} 
\end{figure}

Given logic programs $P_1$ and $P_2$ over the same alphabet, we write
$P_1\approx_n P_2$ if, consistently replacing $SLD$-derivation and
$SLD$-refutation by non-deterministic derivation and nondeterministic refutation
in Definitions~\ref{df:observe} and~\ref{df:observeq}, $P_1$ and $P_2$
are observationally equivalent.

\begin{theorem}[Full abstraction]\label{thm:fullabstract}
  For any logic programs $P_1$ and $P_2$ with the same alphabet, $P_1
  \approx_n P_2$ if and only if ${\bar p_1}\cong {\bar p_2}$.
 \end{theorem}

\begin{proof}
This is routine: as we have allowed any choice of atom rather than
depending upon an algorithm to choose a selected atom, observational
equivalence accounts for all branches.

$\Box$
\end{proof}

The way in which coalgebra models nondeterministic derivations is necessarily complex for a few reasons:

\begin{enumerate}
	\item a non-deterministic derivation might involve infinitely many variables, but each $At(n)$ only allows for a finite number of variables.

\item a non-deterministic derivation could involve an infinite chain of substitutions, but an element of $At(n)$ does not allow for that. Consider e.g. Example \ref{ex:lp3}.
\end{enumerate}

So, within coalgebra, one can only give a chain of finite approximants to a nondeterministic derivation.
Theorem~\ref{th:sc} extends routinely from $SLD$-refutations to non-deterministic
refutations. We can further extend it to non-deterministic derivations too,
with due care for the possibility of derivations involving infinitely many
variables as induced by Example~\ref{ex:lp}.

\begin{theorem}[Soundness and Completeness of non-deterministic
  derivations]\label{th:sc2}
  Let $P$ be a logic program, with $p$ its induced
  $Lax(\ls^{op},P_cP_f)$-coalgebra, and let $G$ be an atomic goal.
\begin{enumerate}
\item Soundness. Given any finite subderivation of a non-deterministic
  derivation of $P\cup \{G\}$ with partial answer $\theta$, the
  subderivation generates a derivation subtree of ${\bar p}(G\theta)$
  for some $n$ with $G\theta \in At(n)$.
\item Completeness. Given a list $\theta_0,\theta_1,\ldots$ of
  substitutions, and a list $T_0,T_1,\ldots$ of finite derivation
  subtrees of ${\bar p}(G\theta_0)$, ${\bar p}(G\theta_0\theta_1)$,
  etcetera, with $T_n\theta_n$ a subtree of $T_{n+1}$ for each $n$,
  there is a non-deterministic derivation of $P\cup \{G\}$ that generates
  the $T_n$'s.
\end{enumerate}
\end{theorem}

\begin{proof}
The soundness claim follows from induction on the length of
a finite subderivation. For length $0$, the statement is trivial.
Assume it is true for length $n$, with derivation subtree
$T_n$ of ${\bar p}(G\theta)$. Suppose $G_{n+1}$ is derived
from $G_n$ using $\theta_{n+1}$ and clause $C_{n+1}$, with respect to the
atom $A_m$ in $G_n$. Apply $\theta_{n+1}$
to the whole of $T_n$, yielding a derivation subtree of
${\bar p}(G\theta\theta_{n+1})$, and extend the tree at the leaf
$A_m\theta_{n+1}$ by applying $\theta_{n+1}$ to each atom in the body of
the $C_{n+1}$ to provide the requisite $and$-nodes.

Completeness also holds by induction. For $n = 0$, given a finite derivation
subtree $T_0$ of ${\bar p}(G\theta_0)$, if follows from the finiteness of $T_0$
and the fact that it is a subtree of ${\bar p}(G\theta_0)$ that it can be
built from a finite sequence of derivation steps starting from $G$, followed
by a substitution.

Now assume that is the case for $T_n$, and we are
given $T_{n+1}$ subject to the conditions stated in the theorem. By our
inductive hypothesis, we have a finite derivation from $G$, followed by
a substitution, that yields the tree $T_n$. That is therefore also true
for $T_n\theta_{n+1}$. As $T_{n+1}$ is a finite extension of $T_n$ and is a subtree
of ${\bar p}(G\theta_0\ldots \theta_{n+1})$, one can make a finite extension
of the finite derivation from $G$ that, followed by a substitution, yields
$T_{n+1}$.
$\Box$
\end{proof}

\section{Corecursion in Logic Programming}\label{sec:corec}

We now move from abstract theory towards the development of
coalgebraic logic programming. Central to this is the relationship
between the finite and the infinite. We introduce a new kind of tree
in order to make the subtle relationship precise and underpin our
formulation of CoALP, a variant of logic programming based on our
coalgebraic semantics.

\subsection{Coinductive derivations}\label{sec:coindtree}

We first return to our running example of program \texttt{Stream}. In Section \ref{sec:lp1} and Figure \ref{pic:stream-n}, we have seen that this program gives rise to non-terminating
SLD-derivations and infinite SLD-trees; moreover, the conventional greatest fixed point semantics is unsound for such cases. 
Coalgebraic semantics of Section \ref{sec:OS} suggests the following tree-based semantics of derivations in \texttt{Stream}, see Figure \ref{pic:stream}.
Comparing Figure \ref{pic:stream-n} and Figure \ref{pic:stream}, we see that computations described by $\&\vee_c$-$Tree$s
suggest parallel branching, much like and-or parallel trees \cite{GPACH12}, and also -- finite height trees in the case of \texttt{Stream}.
These two features will guide us in this Section, when we develop the computational algorithms for CoALP, and then follow them with implementation in Section \ref{sec:impl}. 

We suggest the following definition of \emph{coinductive tree} as a close computational counterpart of the $\&\vee_c$-$Tree$s of the previous section. 
\begin{definition}\label{df:coindt}
Let $P$ be a logic program and $G= A$ be an atomic goal.
The \emph{coinductive tree} for $A$ is
  a possibly infinite tree $T$ satisfying the following properties.
\begin{itemize}
\item $A$ is the root of $T$.
\item Each node in $T$ is either an and-node or an or-node.
\item Each or-node is given by $\bullet$.
\item Each and-node is an atom.
\item For every and-node $A'$ occurring in $T$, if there exist exactly $m>0$ 
distinct  clauses $C_1, \ldots , C_m$ in $P$ (a clause $C_i$ has the form $B_i
  \gets B^i_1, \ldots , B^i_{n_i}$, for some $n_i$),  such that $A' = B_1\theta_1 =
  ... = B_m\theta_m$, for mgu's $\theta_1, \ldots , \theta_m$,  then $A'$ has exactly $m$ children given by
  or-nodes, such that, for every $i \in m$, the $i$th or-node has $n_i$
  children given by and-nodes $B^i_1\theta_i, \ldots ,B^i_{n_i}\theta_i$.
\end{itemize}
\end{definition}



\begin{figure}
\begin{center}
  \begin{tikzpicture}[scale=0.8,baseline=(current bounding box.north),grow=down,level distance=8mm,sibling distance=15mm,font=\footnotesize]
  \node {$\mathtt{stream(x)}$}; 
  \end{tikzpicture}
 $\stackrel{\theta_1}{\rightarrow}$
\begin{tikzpicture}[scale=0.8,baseline=(current bounding box.north),grow=down,level distance=8mm,sibling distance=25mm,font=\footnotesize]
  \node {$\mathtt{stream(scons(z,y))}$}
   child {[fill] circle (2pt)
     child { node {$\mathtt{bit(z)}$}}
     child { node {$\mathtt{stream(y)}$}
     }}; 
  \end{tikzpicture}
	 $\stackrel{\theta_2}{\rightarrow} \ldots \stackrel{\theta_3}{\rightarrow} $
\begin{tikzpicture}[scale=0.8,baseline=(current bounding box.north),grow=down,level distance=8mm,sibling distance=30mm,font=\footnotesize]
  \node {$\mathtt{stream(scons(0,scons(y_1,z_1)))}$}
   child {[fill] circle (2pt)
     child { node {$\mathtt{bit(0)}$}
		   child {[fill] circle (2pt)
			   child {node {$\Box$}}}}
          child { node {$\mathtt{stream(scons(y_1,z_1))}$}
		        child {[fill] circle (2pt)
                child {node {$\mathtt{bit(y_1)}$}}
								child {node {$\mathtt{stream(z_1)}$}}}}   }; 
  \end{tikzpicture}
\end{center}
\caption{\footnotesize{According to the coalgebraic semantics of the previous section, the left hand tree depicts
  ${ \bar p}(\mathtt{stream(x)})$, the middle tree 
depicts ${\bar p}\ At(scons)(\mathtt{stream(x)})$, 
 and the right tree
depicts ${ \bar p}\ At(scons)At(0)At(scons)(\mathtt{stream(x)})$.
The same three trees represent a coinductive derivation for the goal $G = \texttt{stream(x)}$ and the program \texttt{Stream}, with $\theta_1 = x/scons(z,y)$, 
$\theta_2 = z/0$ and$\theta_3 = y/scons(y_1,z_1)$.}}
\label{pic:stream} 
\end{figure}
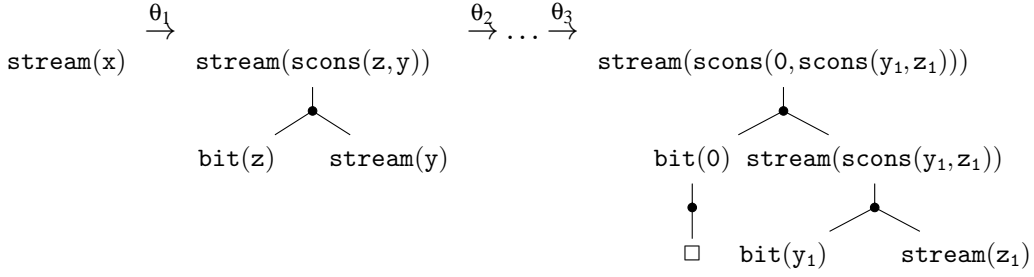

Three coinductive trees for program \texttt{Stream} are shown in Figure~\ref{pic:stream}. 
In contrast to $SLD$-trees, coinductive trees restrict
unification to term matching, i.e., we have $A' = B\theta$, rather
than $A'\theta = B\theta$.  Unification in general is inherently
sequential~\cite{DKM84}, but term matching is parallelisable. 
At the same time,  this
restriction provides a powerful tool for implicit resource
control as it allows one to unfold coinductive trees lazily, keeping
each individual tree at a finite size, 
provided the program is
well-founded; 
as we discuss in detail in Section \ref{sec:guard}.
In our implementation,  we assume that every branch of the coinductive tree can by 
constructed in parallel to other branches, that is, no extra algorithm coordinating the variable substitutions is needed.
See also Sections \ref{sec:parallel} and \ref{sec:impl}.



As can be seen from Figures \ref{pic:tree}  and \ref{pic:stream}, one coinductive tree $T$ may not 
produce the answer corresponding to a refutation by the SLD-resolution. Instead, a sequence of coinductive 
trees may be needed to advance the derivation. 
   We introduce a new derivation algorithm that allows
proof search using coinductive trees.  We modify the definition of
a goal by taking it to be a pair $<A,T>$, where $A$ is an atom, and
$T$ is the coinductive tree determined by $A$. 

\begin{definition}\label{df:coind-res2}
  Let $G$ be a goal given by an atom $\gets A$ and the 
  coinductive tree $T$ induced by $A$, and let $C$ be a clause $H \gets
  B_1, \ldots , B_n$.  Then goal $G'$ is \emph{coinductively derived}
  from $G$ and $C$ using mgu $\theta$ if the following conditions
  hold:
\item[$\bullet$] $A'$ is an atom in $T$.
\item[$\bullet$] $\theta$ is an \emph{mgu} of $A'$ and $H$.
\item[$\bullet$] $G'$ is given by the atom $\gets A\theta$ and the
  coinductive tree $T\theta$. 
\end{definition}

Coinductive derivations resemble \emph{tree rewriting}. They produce the ``lazy'' corecursive effect:  derivations are given by potentially infinite number of steps, where each individual step is executed in finite time.
 
\begin{example}\label{ex:stream2}
	Figure \ref{pic:stream-n} shows how  \texttt{Stream} gives rise to infinite
  $SLD$-trees. But it only gives rise to finite
  coinductive trees because of the term matching condition in the
  definition of coinductive tree.
	Moreover, there is only one coinductive tree for any goal.
  An infinite derivation can be modelled by an infinite 
  coinductive derivation, as illustrated in
  Figure~\ref{pic:stream}.
\end{example}


\begin{example}\label{ex:ct-list}
  \texttt{ListNat}, i.e., Example~\ref{ex:listnat}, also gives rise to
  infinite $SLD$-trees, see Figure~\ref{fig:SLD1}, but it also
  only gives rise to finite coinductive trees as, again, all clauses
  in \texttt{ListNat} are guarded by constructors $\mathtt{0, s, nil, cons}$. A coinductive derivation for \texttt{ListNat} and 
	the goal 
	$\mathtt{list(cons(x,cons(y,x)))}$ is illustrated in Figure~\ref{pic:tree}.
		Again, an infinite derivation can
  be modelled by an infinite chain of finite coinductive trees.
\end{example}

Note that the definition of coinductive derivation allows for non-deterministic choice of the leaf atoms;  compare e.g. with previously seen 
non-deterministic derivations from Definition \ref{df:coind-res}.
  Transitions between coinductive trees can be done in a sequential or 
  parallel manner. That is, if there are several non-empty leaves in
  a tree, any such leaf can be unified with some clause in $P$. Such
  leaves can provide substitutions for sequential or parallel tree
  transitions. In Figure \ref{pic:stream}, the substitution
  $\theta' = \theta_2 \theta_3$ is derived by considering mgu's for two leaves in $G_1 =
  <\texttt{stream(scons(z,y))},T_1>$; but, although two separate and-leaves were used
  to compute $\theta'$, $\theta'$ was computed by composing the two substitutions sequentially, and only
  one tree, $T_3$, was produced. However, we could concurrently derive two trees from $T_2$
  instead, $G'_2 = <\texttt{stream(scons(0,y))},T_2>$ and 
	$G''_2 = <\texttt{stream(scons(z,scons($y_1$,$z_1$)))},T'_2>$. We  exploit parallelism of such transitions in Sections \ref{sec:parallel} and \ref{sec:impl}.

\begin{definition}\label{df:derivsub2}
  Let $P$ be a logic program, $G$ be an atomic goal, and $T$ be a
  coinductive tree  determined by $P$ and $G$.  A subtree
  $T'$ of $T$ is called a \emph{coinductive subtree} of $T$ if it
  satisfies the following conditions:
\begin{itemize}
\item the root of $T'$ is the root of $T$ (up to variable renaming);
\item if an and-node belongs to $T'$, then one of its children belongs
  to $T'$.
\item if an or-node belongs to $T'$, then all its children belong to $T'$.
\end{itemize}
A finite coinductive (sub)tree is called a {\em success (sub)tree} if its leaves are empty goals (equivalently, they are followed only by $\Box$ in the usual
pictures). 
\end{definition}

Note that coinductive subtrees are not themselves coinductive trees: coinductive trees give account to all possible and-or-parallel proof choices given the 
terms determined by the goal, whereas a coinductive subtree corresponds to one possible sequential SLD-derivation for the given goal, where unification in the SLD-derivation 
is restricted to term-matching, cf. Definition~\ref{df:coindt}. 

In what follows, we will assume that the goal in Definition~\ref{df:coind-res2} is given by an atom  $\gets A$, and $T$ is implicitly
assumed. 
This convention agrees with the standard logic programming practice, where goals are given by first-order atoms.
For example,  we  say that the goal \texttt{stream(x)} generates the coinductive derivation of Figure~\ref{pic:stream}.
The next definition formalises this convention.  


\begin{definition}
A \emph{coinductive derivation} of $P\cup \{G\}$ consists of a
sequence of goals $G=G_0, G_1, \ldots$ called \emph{coinductive resolvents} and a sequence
$\theta_1,\theta_2,\ldots$ of mgu's such that each $G_{i+1}$ is
derived from $G_i$ using $\theta_{i+1}$. A
\emph{coinductive refutation} of $P \cup \{G\}$ is a finite
coinductive derivation of $P \cup \{G\}$ such that its last tree contains a success subtree. 
If $G_n$ contains a success subtree, we say
that the refutation has length $n$.
\end{definition}

	
	We now modify Definitions \ref{df:observe} and \ref{df:observeq} of observational equivalence.
	Suppose the definitions of  a (correct) call pattern and a (correct) partial answer from Definition~\ref{df:observe} are re-formulated
	with respect to coinductive derivations, rather than SLD-derivations. Unlike SLD-derivations, coinductive derivations perform computations in 
	``two dimensions'' -- at the level of coinductive trees and at the level of transitions between coinductive trees.  
Both dimensions of computations can be observed. The next definition formalises this.
	

\begin{definition}\label{df:observeq2}
  Let $P_1$ and $P_2$ be logic programs with the same alphabet. Then
  $P_1$ is {\em coinductively observationally equivalent} to $P_2$, written $P_1
  \approx_c P_2$, if for any goal $G$, the following conditions
  hold:
\begin{itemize}
\item[1.-4.] Conditions of Definition ~\ref{df:observeq}, but with coinductive derivations
replacing SLD-derivations.
\item[5.] The coinductive tree $T_1$ for $G$ and $P_1$ contains a coinductive subtree $C$ iff the coinductive tree $T_2$ for $G$ and $P_2$ contains $C$, modulo variable renaming.
\end{itemize}

\end{definition}

For ground programs, all coinductive derivations will have length $0$, and the coinductive tree generated for a given goal will account for
 all alternative derivations by SLD-resolution. 
Therefore, conditions [1.-4.] of coinductive observational equivalence will be trivially satisfied for all ground logic programs.
However, condition [5.] will be able to distinguish different logic programs in such cases.

\begin{theorem}[Full abstraction]\label{thm:fullabstract2}
  For any logic programs $P_1$ and $P_2$ with the same alphabet, $P_1
  \approx_c P_2$ if and only if ${\bar p_1}\cong {\bar p_2}$.
 \end{theorem}
\begin{proof}

Similarly to Theorem \ref{thm:fullabstract}, we allowed any choice of resolvents, and observational
equivalence accounts for all branches. This accounts for conditions [1.-4.] in Definition \ref{df:observeq2}.

For condition [5.] of coinductive observational equivalence, consider coinductive trees: their structure and labels account for all possible clauses that can be \emph{matched}
with the current goal and subgoals via mgu's. If, for any goal $G$ with $n$ distinct variables, $P_1$ and $P_2$ produce equivalent coinductive trees, then the image of $G$ under ${\bar p_1}$
will be isomorphic to the image of $G$ under ${\bar p_2}$.

The other direction is straightforward.
$\Box$
\end{proof}

In general, the definition of the coinductive tree  permits generation of coinductive trees containing
infinitely many variables. So a coinductive tree for a goal
$A$ need not be a subtree of ${\bar p}(A)\in \&\vee_c$-$Tree_{At(n)}$
for any $n$. But every finite one must be, and establishment or
otherwise of finiteness of coinductive trees is critical
for us.

\begin{example}\label{ex:ct-GC}
  \texttt{GC}, i.e., Example~\ref{ex:lp}, has a clause
\begin{eqnarray*}
  \texttt{connected(x,y)} \gets
  \texttt{edge(x,z)},\texttt{connected(z,y)}
\end{eqnarray*}
in which there is a variable in the body but not the head. If one
includes a unary function symbol $s$ in \texttt{GC}, the clause
induces infinite coinductive trees, all subtrees of\\
 ${\bar   p}(\texttt{connected(x,y)})\in \&\vee_c$-$Tree_{At(2)}$, as there are
infinitely many possible substitutions for $z$. The clause also induces
infinitely many coinductive trees that do not lie in ${\bar
  p}(\texttt{connected(x,y)})\in \&\vee_c$-$Tree_{At(n)}$ for any
$n$.
\end{example}

Note that, in Section \ref{sec:OS}, we established two different kinds of soundness and completeness results: one related the coalgebraic semantics to 
finite SLD-refutations (cf Theorem \ref{th:sc}), another -- to potentially infinite non-deterministic derivations (cf Theorem \ref{th:sc2}).
The second theorem subsumed the first for special cases of logic programs.
As we explain in the next section, one of the main advantages of CoALP is graceful handling of corecursive programs and coinductive definitions.
This is why we consider derivations of arbitrary size  in our next statement of soundness and completeness for CoALP, as follows.

\begin{theorem}[Soundness and Completeness of coinductive
  derivations]\label{th:sc3}
  Let $P$ be a logic program, with $\bar{p}$ its induced
  $Lax(\ls^{op},P_cP_f)$-coalgebra, and let $G$ be an atomic goal.
\begin{enumerate}
\item Soundness. Given a coinductive tree $\mathcal{T}$ resulting from a coinductive
  derivation of $P\cup \{G\}$ with partial answer $\theta$, there is a coinductive subtree $C$ of $\mathcal{T}$,
  such
that $C$ is a derivation subtree of $\bar p(G\theta)$ for some $n$ for which $G\theta \in
At(n)$.
	
\item	Completeness. Given a list $\theta_0,\theta_1,\ldots$ of
  substitutions, and a list $T_0,T_1,\ldots$ of finite derivation
  subtrees of ${\bar p}(G\theta_0)$, ${\bar p}(G\theta_0\theta_1)$,
  etcetera, with $T_n\theta_n$ a subtree of $T_{n+1}$ for each $n$,
  there is a coinductive derivation of $P\cup \{G\}$, involving computed substitutions $\theta_0, \theta_1, \ldots$ and
	coinductive trees $\mathcal{T}_0, \mathcal{T}_1, \ldots$ such that, for each $n$, $\mathcal{T}_n$ contains a 
	coinductive subtree $C_n$, such that $C_n$  contains $T_n$, modulo variable renaming. 
	
	

\end{enumerate}
\end{theorem}
\begin{proof}
Soundness. Consider a coinductive derivation of $P\cup \{G\}$ with partial answer $\theta = \theta_0, \theta_1, \ldots , \theta_k$:
it contains a sequence of coinductive trees $\mathcal{T}_0, \mathcal{T}_1, \ldots , (\mathcal{T}_k =\mathcal{T})$ for\\ $G\theta_0, G\theta_1, \ldots , G\theta_k$. 
Each $\mathcal{T}_i$ is uniquely determined by $G\theta_i$, although $\mathcal{T}_i$  may have infinite branches (cf. Example \ref{ex:lp3}). 
In general case, $\mathcal{T}$ may contain several coinductive subtrees, each giving an account to one possible combination of clauses determining or-nodes.
Consider one such coinductive subtree $C$, and suppose it contains $n$ distinct variables. Then, by construction of $\&\vee_c$-$Tree_{At(n)}$ and Definition \ref{df:derivsub}, there will be a derivation subtree
in $\bar p(G\theta)$ corresponding to $C$.

Completeness. 
The proof is similar to the proof of Theorem \ref{th:sc2}, but here, we also note that each step in a coinductive derivation is given by a coinductive tree, rather than by a resolvent.
The role of a non-deterministic SLD-derivation is now delegated to a coinductive subtree $C_n$ of the coinductive tree $\mathcal{T}_n$.
Note that coinductive trees may be finite for guarded clauses like \texttt{Stream} (cf. Example \ref{ex:stream2}), and hence a sequence of coinductive trees $\mathcal{T}_1, \mathcal{T}_2, \ldots$ will yield all $T_n$\rq{}s (cf. Theorem \ref{th:sc2}). However, non-guarded clauses (cf. Example \ref{ex:lp3}) give rise to infinite coinductive trees, in which case $T_n$ will be only a fragment of a coinductive subtree $C_n$ of the coinductive  tree $\mathcal{T}_n$. In that case,  an infinite sequence of $T_n$s would approximate one $C_n$, similarly to Theorem \ref{th:sc2}.  
$\Box$
\end{proof}

Discussion of the constructive componenent of the completeness results for CoALP and the constructive reformulation of the above completeness theorem can be found in \cite{KSH13}.
The problem of distinguishing cases with finite and infinite coinductive trees will be the main topic of the next section.

\subsection{Guarding Corecursion}\label{sec:guard}
In this section, we consider various methods used in logic programming to \emph{guard} (co-) recursion,
and introduce our own method for guarding corecursion in CoALP.

As Example \ref{ex:stream1} illustrates, $SLD$-derivations may yield
looping infinite derivation chains for programs like stream. In \emph{Coinductive Logic Programming} (Co-LP)~\cite{GuptaBMSM07,SimonBMG07},
such were addressed by introducing a procedure allowing one to terminate derivations with the flag
``\texttt{stream(x)} is proven'',  whenever
such a loop was detected.  Extending this ``rule'' to
inductive computations leads to unsound results: in the inductive
case, infinite loops normally indicate lack of progress in a
derivation rather than ``success''. Thus, explicit annotation of
predicates was required. Consider the following example.

\begin{example}\label{ex:stream3}
The annotated logic program below comprises both inductive and coinductive clauses.
\vspace{-10pt}
\small{
\begin{eqnarray*}
\texttt{bit}^{i}(0) & \gets &\\
\texttt{bit}^i(1) & \gets & \\
\texttt{stream}^c(\texttt{scons} (x,y)) & \gets &  \texttt{bit}^i(x),\texttt{stream}^c(y) \\
\texttt{list}^i(\texttt{nil}) & \gets & \\
\texttt{list}^i(\texttt{cons} (x,y)) & \gets & \texttt{bit}^i(x),\texttt{list}^i(y) \\
\end{eqnarray*}
}
%
Only infinite loops produced for corecursive goals (marked by $c$) are
gracefully terminated; others are treated as ``undecided''
proof branches.
\end{example}

In practice, these annotations act as locks and keys in resource logics,
allowing or disallowing infinite data structures.  There are
drawbacks:

\begin{itemize}
\item[$\bigstar$] some predicates may behave inductively or coinductively
depending on the arguments provided, and such cases need to be
resolved dynamically, not statically, in which case predicate
annotation fails.

\item[$\bigstar\bigstar$] the coinductive algorithm~\cite{GuptaBMSM07,SimonBMG07} is not in essence a
lazy infinite (corecursive) computation. Instead, it substitutes an
infinite proof by a finite derivation, on the basis of guarantees of
the data regularity in the corecursive loops.  But such guarantees
cannot always be given: consider computing the number $\pi$.
\end{itemize}

The coinductive derivations we introduced in the previous section give an alternative solution to the problem of guarding corecursion. 
We have already seen that Definition~\ref{df:coindt} determined finite coinductive trees both 
for the coinductive program \texttt{Stream} and inductive program \texttt{ListNat}; and no explicit annotation was needed to handle this. 
These two programs were \emph{well-founded}, however, not all programs will
give rise to finite coinductive trees. 
This leads us to the following definition of well-foundness of logic programs.

  
\begin{definition}
A logic program $P$ is well-founded if, for any goal $G$, 
 $P\cup \{G\}$ generates the coinductive tree of finite size.
	\end{definition}

There are logic programs that allow infinite coinductive trees.
\begin{example}\label{ex:ct-R}
  Consider $R(x) \gets R(f(x))$. The coinductive tree arising from
  this program contains a chain of alternating $\bullet$'s and atoms
  $R(x)$, $R(f(x))$, $R(f(f(x)))$, etcetera, yielding an infinite coinductive
  tree. This tree is a subtree of ${\bar p}(R(x))\in
  \&\vee_c$-$Tree_{At(1)}$.
\end{example}	




	
	In line with the existing practice of functional languages, we want the notion of well-foundness to be transformed into programming practices.
	For this, a set of syntactic \emph{guardedness} conditions needs to be introduced, compare e.g. with \cite{BK08,Coq94,Gimenez98}. 
	Coinductive trees we introduced in the previous section allow us to  
	formulate similar guardedness conditions. 
They correspond to the method of 
\textbf{guarding} (co)recursive function applications \textbf{by constructors} in \cite{Coq94,Gimenez98}. In our running examples, function
symbols \texttt{0}, \texttt{1}, \texttt{s}, \texttt{cons}, \texttt{scons}, \texttt{f} play the role of guarding constructors.




\textbf{Guardedness check 1 (Presence of Constructors):} 
\emph{If a clause has the form $P(\bar{t}) \gets [atoms], P(\bar{t'}), [atoms]$, where $P$ is a predicate, $\bar{t}$, $\bar{t'}$ are  lists of terms,
and $[atoms]$ are finite (possibly empty) lists of first-order atoms, 
then at least one term $t_i \in \bar{t}$  must contain  a function symbol $f$.}


For example, \texttt{Stream} is guarded. But Check-1 is not sufficient to guarantee well-foundness of coinductive trees. Consider the following examples. 


\begin{example}\label{ex:ct-R1}
  Consider the variant of Example~\ref{ex:ct-R} given by $R(f(x))
  \gets R(f(f(x)))$. It generates an infinite
  coinductive tree, given by a chain of alternating $\bullet$'s and
  atoms $R(f(x))$, $R(f(f(x)))$, etcetera.
\end{example}

\begin{example}[\texttt{Stream2}]
Another non-well-founded program that satisfies Guardedness check 1 is given below:\\
\footnotesize{
\texttt{stream2(scons(x,y))} $\gets$ \texttt{bit(x), stream2(scons(x,y))}}
\end{example} 

To address such problems, a second guarding condition is needed. 

\textbf{Guardedness check 2 (Constructor Reduction):}
\emph{If a clause has the form $P(\bar{t}) \gets  [body]$, where $P$ is an $n$-ary predicate, $\bar{t}$ is a list of terms $t_1, \ldots , t_n$, and $[body]$ is a finite non-empty list of first-order atoms,
then, for each occurence of $P(\bar{t'})$ (with some $\bar{t'} = t_1', \ldots , t_n'$) in $[body]$, the following conditions must be satisfied.
There should exist a  term $t_i \in \bar{t}$ such that, there is a function $f$ that occurs in $t_i$ $m$ times ($m\geq 1$) and occurs in $t'_i$ \ $k$ times  with $k<m$.
Moreover, if $f \in t_i$ has arguments containing variables $\bar{x_i}$, then $f \in t_i\rq{}$ must have arguments containing variables $\bar{x_i\rq{}}$, with $\bar{x_i\rq{}} \subseteq \bar{x_i}$; 
if $f$ occurs in $t_i$ but not in $t_i'$, then all  variables $\bar{x_i\rq{}} \in t_i'$ must satisfy $\bar{x_i\rq{}} \subseteq \bar{x_i}$.
}

   \begin{example}
   Suppose we want to define a program that computes and infinite stream
of natural numbers:
$0,1,2,3,4,5,...$

The corresponding logic program will be given
by\\ 
{\footnotesize{\texttt{nats(scons(x,scons(y,z)))} $\gets$ \texttt{nat(x), nat(s(x)), nats(scons(y,z))}}}

It is a well-founded and guarded program, so will result in potentially infinite coinductive derivations featuring
 coinductive trees of finite size. 
This program will satisfy Guardedness conditions 1 and 2: the function symbol (constructor) \texttt{scons} reduces in the body.
   \end{example}  

\begin{example}
In Example \ref{ex:ct-R1}, function symbol $f$ appears twice in the body, while appearing only once in the head; this breaks the guardedness condition 2.
\end{example}

Note that Guardedness check 2 imposes strict discipline on argument positions at which constructors reduce, and on variables appearing as arguments to the constructors. The next example explains why these restrictions matter.

\begin{example}
Consider the following clause:\\
\texttt{Q(s(x),y)} $\gets$ \texttt{Q(y,y)}.

The constructor \texttt{s} clearly reduces, and the clause could pass the guardedness check if it was checking only the constructor reduction. However, the goal $Q(s(x),s(x))$ would result in an infinite coinductive tree.
The problem here is the new variable \texttt{y} in the body, in the same argument position as \texttt{s(x)}: it allows to form the goals like $Q(s(x),s(x))$ falling into infinite loops.
To avoid such cases, Guardedness check 2 imposes the restriction on the argument positions and variables. Therefore, the programmer would be forced to change the clause to\\
\texttt{Q(s(x),y)} $\gets$ \texttt{Q(x,y)}\\
to pass the guardedness checks.

\end{example}

Finally, the (co-)recursive nature of the predicates may show only via several clauses in the program. Consider the following example.

\begin{example}
Consider programs
\texttt{P1} and \texttt{P2} below. For both programs, Guardedness conditions 1 and 2 are satisfied for every single clause, but the programs 
give rise to infinite coinductive trees. 



\texttt{P1}:\\
\texttt{ Q(cons(x,y))} $\gets$ \texttt{Q2(cons(z,cons(x,y)))}.\\
\texttt{ Q2(cons(z,cons(x,y))} $\gets$ \texttt{Q(cons(x,y))}.

\texttt{P2}:\\
 \texttt{Q(cons(x,y))} $\gets$ \texttt{Q2(cons(z,cons(x,y)))}.\\
 \texttt{Q2(y)} $\gets$ \texttt{Q(y)}.

\end{example}

To address the problem above, a further guardedness check needs to be introduced.

\begin{definition}
Given  a logic program $P$, a goal $G$, and the coinductive tree for $P\cup \{G\}$, we say \emph{$T$ contains a loop} if there exists a coinductive subtree $C$ of $T$, such that:\\
there exists a predicate $Q \in P$ such that  $Q(\bar{t})$ appears as an and-node in $C$, and also  $Q(\bar{t'})$ appears as a child and-node of that node in $C$, for some  $\bar{t}$ and $\bar{t'}$.

In this case, we say that atom $Q(\bar{t})$ is a head loop factor, and $Q(\bar{t'})$ is a tail loop factor.
\end{definition}

\textbf{Guardedness check 3 (Detection of Non-guarded Loops):} \emph{If a program $P$ satisfies guardedness conditions 1 and 2, do the following.
For every clause $C \in P$, such that $C$ has the shape $A \gets B_1, \ldots B_n$, 
construct the coinductive tree $T$ for $A$, imposing the following termination conditions during the tree construction:}

\begin{itemize}
	\item[\emph{i.}] If $T$ contains a loop with the head and tail factors $Q(\bar{t})$ and $Q(\bar{t'})$, apply Guardedness 
	conditions \textbf{1} and \textbf{2} to  $Q(\bar{t}) \gets Q(\bar{t'})$.
	If the Guardedness conditions \textbf{1} and \textbf{2} are violated for $Q(\bar{t}) \gets Q(\bar{t'})$, terminate the coinductive tree construction for $A$; report non-guardedness.
	\item[\emph{ii.}] If construction of $T$ reaches the leaf nodes and none of the guardedness conditions \emph{(i.)} and \emph{(ii.)} is violated, the program $P$ is guarded.
\end{itemize}

\begin{proposition}
Guardedness check 3 terminates, for any logic program.
\end{proposition}
\begin{proof}
Note that a given program $P$ has a finite and fixed number of clauses. If there are $n$ clauses in the given program, only $n$ coinductive trees will be constructed. It remains to show that each tree construction will be terminated in finite time. Given that $P$ contains a finite number of predicates, an infinite coinductive tree $T$ for $P$ would need to contain a loop.
If all loops occuring in $T$ are guarded, they could not have constructor reduction infinite number of times, so there  should be at least one non-guarded loop. 
But then the tree construction will be terminated, by item $i$. $\Box$ 
\end{proof}

Note that, although the procedure above requires some computations to be performed, the guardedness checks can be done statically, prior to the program run. 
 
\begin{example}
Consider the program \texttt{P3}:\\
\texttt{Q(cons(x,y)) }$\gets$ \texttt{Q2(cons(z,cons(x,y))}. \\
\texttt{Q2(cons(cons(x,y))}  $\gets$ \texttt{Q(y)}.

It satisfies guardedness checks \textbf{1}, \textbf{2} and \textbf{3}. In particular, coinductive trees for both of its clauses are finite, and show constructor reduction.
\end{example}

The Guardedness checks 1-3 are necessary, but not sufficient conditions for guaranteeing well-foundness of all logic programs. This is why, we include some further checks, involving applying checks 1-3 to program heads modulo some chosen substitutions. We will not go into further details here, but will illustrate the issue by the following example.

  \begin{example}
  Consider the logic program \texttt{P4}:\\
  \texttt{Q(s(x),y)} $\gets$ \texttt{P(x,y)}\\
  \texttt{P(t(x),y)} $\gets$ \texttt{Q(y,y)}
  
  Each clause passes the Guardedness checks 1-2 trivially, as they do not have immediate loops. When we construct coinductive trees for each of the clause heads, they do not exhibit the loops, either, due to the restrictive nature of the term matching. However, for the goal \texttt{Q(s(t(x)),s(t(x)))}, the program will give rise to an infinite coinductive tree. 
  \end{example}

Guardedness conditions of CoALP guarantee that, if a program $P$ passed the guardedness checks, then any goal will give rise to only finite coinductive trees. 
Very often, in functional programming, the guardedness conditions reject some well-founded programs \cite{BK08,Coq94,Gimenez98}. 
Termination of recursive programs is in general undecidable, and syntactic guardedness conditions are used only to approximate the notion of termination. 

Here, as well as in functional programming, there will be examples of well-founded but non-guarded programs:

\begin{example}
The Program \texttt{P5} is well-founded but not guarded:\\
\texttt{Q(s(x)),y)} $\gets$ \texttt{Q(y,x)} \\
\texttt{Q(x,s(y))} $\gets$ \texttt{Q(y,x)}
\end{example}

Furthermore, the guardedness checks are too restrictive to capture the notion of termination in sequential logic programs as given by e.g. SLD-resolution.  

\begin{example} The following program is non-well-founded and not guarded in CoALP setting, but terminates if SLD-resolution is used:\\
\texttt{Q(x)} $\gets$ \texttt{Q(a)}. \\
\texttt{Q(a)}  $\gets$ .

As we discuss in the next section, the program $GC$ gives a similar effect. 
\end{example}
 




Our approach allows us to guard (co-)recursion implicitly, without annotating the predicates as inductive or coinductive, as it was the case in \cite{GuptaBMSM07,SimonBMG07}. 
The advantages of this implicit method of handling (co-)recursive
computational resources can be summarised as follows.  It solves both
difficulties that explicit coinductive resource management causes:
 in response to $\bigstar$, the method uniformly treats inductive and coinductive definitions, and it can be used to detect non-well-founded cases in both;
 and in response to $\bigstar\bigstar$, it is a corecursive process in
 spirit. Thus, instead of relying on guarantees of loop regularity, it
 relies on well-foundness of every coinductive tree in the
 process of lazy infinite derivations.

\subsection{Programming with Guarded Corecursion}\label{ap:A}
We proceed with a case study of how guardedness conditions can be used in 
logic programming practice.

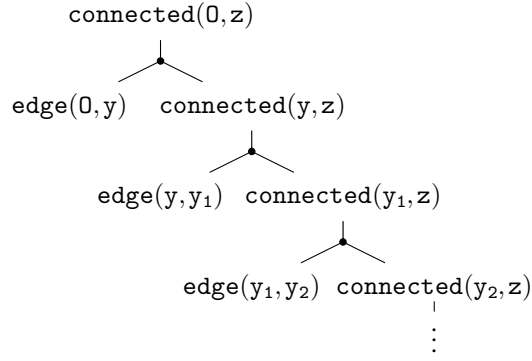
\begin{figure}
\begin{center}
  \begin{tikzpicture}[scale=0.3,baseline=(current bounding box.north),grow=down,level distance=20mm,sibling distance=80mm,font=\footnotesize]
  \node {$\mathtt{connected(O,z)}$}
   child {[fill] circle (4pt)
     child { node {$\mathtt{edge(O,y)}$}}
       child { node {$\mathtt{connected(y,z)}$}
        child {[fill] circle (4pt)
       child { node{$\mathtt{edge(y,y_1)}$}}
         child { node{$\mathtt{connected(y_1,z)}$}
				        child {[fill] circle (4pt)
       child { node{$\mathtt{edge(y_1,y_2)}$}}
         child { node{$\mathtt{connected(y_2,z)}$}
				child {node{$\vdots$}}}}}}}}; 
  \end{tikzpicture}
 \end{center}
\caption{\footnotesize{The infinite coinductive tree for the programs GC from Example~\ref{ex:lp},  GC' from Example~\ref{ex:lp4} and GC* from Example~\ref{ex:lp-cyclic}.
} }
\label{pic:GC*} 
\end{figure}

First, we consider the effects of coalgebraic logic programming on corecursive
resource handling by comparing \texttt{GC} (cf. Example~\ref{ex:lp}) with
\texttt{Stream}. \texttt{GC} uses recursion to
traverse all the connected nodes in a graph.  Two kinds of infinite
$SLD$-derivations are possible: computing finite or infinite objects.

\begin{example}[GC*]\label{ex:lp-cyclic}
  Consider the program GC*. Adding the following  clause to \texttt{GC} makes the graph cyclic:\\
  $\texttt{edge(s(s(0)),0)} \ \gets$.  
	
	Taking a query $\gets
  \texttt{connected(0,y)}$ as a goal may lead to an infinite
  $SLD$-derivation corresponding to an infinite path starting from
  $\texttt{0}$ in the cycle. It would also give rise to infinite coinductive trees, see Figure~\ref{pic:GC*}.  However, the object that is described by
  this program, the cyclic graph with three nodes, is finite.
\end{example}

In the standard practice of logic programming, where the ordering of
the clauses is as in \texttt{GC}, the program behaves gracefully,
giving finitely computed answers, but potentially infinitely many
times. But this balance is fragile. For example, the following
program, with different ordering of the clauses and of the atoms in
the body, results in non-terminating derivations:

\begin{example}[GC']\label{ex:lp4}
Let \texttt{GC'} denote the logic program
\begin{eqnarray*}
1. \  \texttt{connected(x,y)} & \gets & \texttt{connected(z,y)}, \texttt{edge(x,z)}\\
2. \ \texttt{connected(x,x)} & \gets &
\end{eqnarray*}
together with  the database of Example~\ref{ex:lp},
$SLD$-derivation loops as follows: \\
$\texttt{connected(0,$y$)} \rightarrow (\texttt{connected($z$,$y$)}, \texttt{edge(0,$z$)})\rightarrow \\
(\texttt{connected($z_1$,$y$)},\texttt{edge($z$,$z_1$)}, \texttt{edge(0,$z$)})$ $\rightarrow \ldots $\\
It never produces an answer as it falls into an infinite loop
irrespective of the particular graph in question.

There is a one-step non-deterministic derivation for \texttt{connected(0,$y$)} 
given by unifying $y$ with $0$ (see Definition \ref{df:coind-der}.)
But there is no coinductive derivation that does that: see Figure \ref{pic:GC*}.

Spelling out nondeterministic semantics (Theorem \ref{th:sc2}),\\
$T_1 = connected(0,y)$;\\
$T_0 = connected(0,0) \rightarrow \Box$.


\end{example}

\begin{figure}
\begin{center}
\footnotesize{
  \begin{tikzpicture}[scale=0.3,baseline=(current bounding box.north),grow=down,level distance=20mm,sibling distance=80mm,font=\footnotesize]
  \node {$\mathtt{conn(O,cons(y,z))}$}
   child {[fill] circle (4pt)
     child { node {$\mathtt{edge(O,y)}$}}
		child { node {$\mathtt{conn(y,z)}$}
    }}; 
  \end{tikzpicture}
	$\rightarrow$
  \begin{tikzpicture}[scale=0.3,baseline=(current bounding box.north),grow=down,level distance=20mm,sibling distance=90mm,font=\footnotesize]
  \node {$\mathtt{conn(O,cons(sO,z))}$}
   child {[fill] circle (4pt)
     child { node {$\mathtt{edge(0,s0)}$}
		child {[fill] circle (4pt)
		 child { node {$\Box$} }}}
       child { node {$\mathtt{conn(s0,z)}$}
               }}; 
  \end{tikzpicture}
	$\rightarrow$
	 \begin{tikzpicture}[scale=0.3,baseline=(current bounding box.north),grow=down,level distance=20mm,sibling distance=95mm,font=\footnotesize]
  \node {$\mathtt{conn(0,cons(s0,nil))}$}
   child {[fill] circle (4pt)
     child { node {$\mathtt{edge(0,s0)}$}
		child {[fill] circle (4pt)
		 child { node {$\Box$} }}}
       child { node {$\mathtt{conn(s0,nil)}$}
			child {[fill] circle (4pt)
		 child { node {$\Box$} }}
               }}; 
  \end{tikzpicture}}  
\end{center}
\caption{\footnotesize{A finite and well-founded coinductive derivation for a guarded variant of \texttt{$GC^g$}; we use \texttt{conn} to abbreviate \texttt{connected};
  \texttt{s0} abbreviates \texttt{s(0)}. }}
\label{pic:inftree2} 
\end{figure}
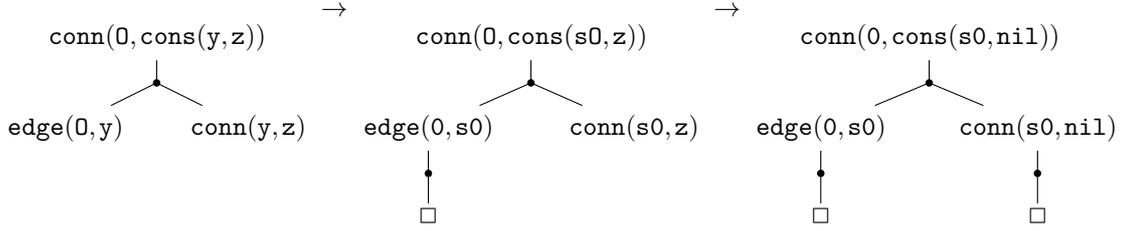

In traditional logic programming, the burden of deciding which
programs might result in loops like the one above falls completely to
the programmer: semantically, \texttt{GC} and \texttt{GC'} are
equivalent. Moreover, in the Co-LP \cite{GuptaBMSM07,SimonBMG07} setting, if the atoms in the programs above are labelled
as inductive, the behaviour of Co-LP is
exactly as it is for $SLD$-resolution. If, on the contrary, the atoms
are marked as coinductive, we may find the derivation loop terminated
as ``successful'' when we should be warned of its being
non-well-founded.

In contrast, compare the coalgebraic semantics of \texttt{GC}, \texttt{GC'}, \texttt{GC*} and
\texttt{Stream}. Figures~\ref{pic:GC*} and~\ref{pic:stream} show the difference between the coinductive trees for 
ill-founded \texttt{GC}, \texttt{GC'} and \texttt{GC*} and  well-founded programs like \texttt{Stream}.
Notably, coinductive definition of \texttt{Stream} is well-founded, while traditional inductive definition of \texttt{GC*} is not.
\texttt{GC}, \texttt{GC'} and \texttt{GC*} give rise to infinite
coinductive trees,
whereas \texttt{Stream} gives rise
only to finite coinductive trees. 

In CoALP, a set of syntactic guardedness checks 1-3 is embedded, to make sure that 
only programs that satisfy the semantic notion of well-foundness are allowed in CoALP.
Programs like GC, GC' and GC*  will be automatically rejected by CoALP's guardedness checks, see Section \ref{sec:impl}. 
To make the programs like \texttt{GC} guarded, 
 The user will have to reformulate it as
follows:

\begin{example}[$GC^g$]
The program $\mathtt{GC^g}$ below addresses both non-terminating problem for SLD-derivations for \texttt{GC'}, and 
non-well-foundness of \texttt{GC} and \texttt{GC*}.

\begin{eqnarray*}
\texttt{connected(x,cons(y,z))} & \gets &  \texttt{edge(x,y),connected(y,z)}\\
\texttt{connected(x,nil)} & \gets & \\
\texttt{edge(0,0)} & \gets & \\
\texttt{edge(x,s(x))} & \gets & \\
\end{eqnarray*}
The coinductive derivation for it is shown in Figure
\ref{pic:inftree2}, duly featuring coinductive trees  of finite size.
 
\end{example}

\section{Guarding Parallelism by Guarded Corecursion}\label{sec:parallel}



One of the distinguishing features of logic programming languages is
that they allow implicit parallel execution of programs.  In the last
two decades, an astonishing variety of parallel logic programming
implementations have been proposed, see \cite{GPACH12} for a detailed
survey.  The three main types of parallelism used in implementations
of logic programs are \emph{and-parallelism},
\emph{or-parallelism} and their
combination; see also Section \ref{sec:trees2}.
The coalgebraic models we discuss in this paper exhibit a synthetic
form of parallelism: and-or parallelism.  The most common way to
express and-or parallelism in logic programs is
and-or parallel derivation trees~\cite{GuptaC94,GPACH12}, see Definition~\ref{df:andortree}.

In the ground case, coinductive trees and and-or parallel
derivation trees agree, as illustrated by Example \ref{ex:lptree}. But
as we have discussed several times, that does not extend.  In the
general case, in the absence of synchronisation, parallel
and-or-trees may lead to unsound results.

\begin{example}\label{ex:unsound}
  Figure \ref{pic:and-or} depicts an and-or parallel derivation tree
  that finds a refutation $\theta = \{x/0, y/0, x/nil\}$ for the goal
  \texttt{list(cons(x,cons(y,x)))}, although this answer is not sound.
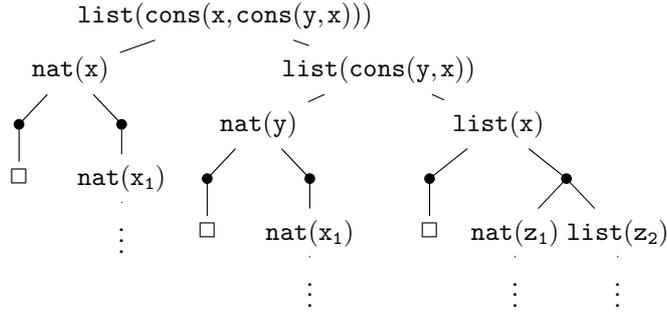
\begin{figure}
\begin{center}
  \begin{tikzpicture}[scale=0.9,baseline=(current bounding box.north),grow=down,level distance=8mm,sibling distance=33mm,font=\footnotesize]
  \node [sibling distance=35mm]{$\mathtt{list(cons(x,cons(y,x)))}$}
    child  [sibling distance=45mm]{ node {$\mathtt{nat(x)}$}
       child  [sibling distance=15mm]{[fill] circle (2pt)
		    child [sibling distance=15mm]{node {$\Box$}}}
				child  [sibling distance=15mm] {[fill] circle (2pt)
				child [sibling distance=15mm]{ node {$\mathtt{nat(x_1)}$}
				    child {node{$\vdots$}}}}}
    child  [sibling distance=45mm] { node {$\mathtt{list(cons(y,x))}$}
       child [sibling distance=35mm]{ node{$\mathtt{nat(y)}$}
			  child [sibling distance=15mm] {[fill] circle (2pt)
		       child [sibling distance=15mm] {node {$\Box$}}}
					child [sibling distance=15mm]{[fill] circle (2pt)
				   child [sibling distance=15mm] { node {$\mathtt{nat(x_1)}$}
				        child [sibling distance=15mm] {node{$\vdots$}}}}}
       child [sibling distance=35mm]{ node{$\mathtt{list(x)}$}
				child [sibling distance=20mm] {[fill] circle (2pt)
		      child [sibling distance=15mm] {node {$\Box$}}}
					child [sibling distance=20mm] {[fill] circle (2pt)
				  child [sibling distance=15mm] { node {$\mathtt{nat(z_1)}$}
				     child [sibling distance=15mm] {node{$\vdots$}}}
				  child [sibling distance=15mm] { node {$\mathtt{list(z_2)}$}
				     child [sibling distance=15mm] {node{$\vdots$}}}}}
		       }; 
  \end{tikzpicture}
\end{center}
\caption{\footnotesize{An unsound refutation by an and-or parallel derivation tree, with $\theta = \{x/0, y/0, x/nil\}$ .}}\label{pic:and-or}
\end{figure}	
\end{example}

A solution proposed in \cite{GuptaC94} was given by \emph{composition
  (and-or parallel derivation) trees}. Construction of composition
trees involves additional algorithms that synchronise substitutions in
the branches of and-or parallel derivation trees.  Composition trees
contain a special kind of composition nodes, used whenever both and-
and or-parallel computations are possible for one goal. A composition
node is a list of atoms in the goal.  If, in a goal $G \ = \ \gets
B_1, \ldots B_n$, an atom $B_i$ is unifiable with $k>1$ clauses, then
the algorithm adds $k$ children (composition nodes) to the node $G$;
similarly for every atom in $G$ that is unifiable with more than one
clause.  Every such composition node has the form $ B_1, \ldots B_n$
and has $n$ and-parallel edges emanating from it. Thus, all possible
combinations of or-choices at every and-parallel step are
given.

Predominantly, the existing parallel implementations of logic programming
 follow Kowalski's principle \cite{Kow79}:
$$\emph{Programs = Logic + Control}.$$

This principle separates the control component (backtracking, occur 
check, goal ordering/selection, parallelisation, variable synchronisation ) 
from the logical specification of a problem (first-order Horn logic, 
$SLD$-resolution, unification). 
Thus the control of program execution becomes independent of 
programming semantics.

With many parallel solutions on offer~\cite{GPACH12}, some form of resource handling
and process scheduling are inevitable ingredients of parallel logic
programming as the algorithms of
unification and $SLD$-resolution are
P-complete \cite{UllmanG88,Kanellakis88} and cannot be
parallelised in general, see Example \ref{ex:unsound}.  Parallel
implementations of PROLOG typically hide all additional control-handling
algorithms at the level of implementation, away from program
specification or semantics \cite{GPACH12}. 
The algorithms used for variable synchronisation pose a sequential barrier for parallelisation.

Several properties 
are shared by many  parallel
implementations of PROLOG:

\begin{itemize}
\item[$\star$] Although and-or-parallelism is called ``implicit parallelism''
in the literature \cite{GPACH12}, it boils down to explicit resource
handling at compiler level: this includes both annotating the syntax
and maintaining special schedulers/arrays/hash tables to synchronise
variable substitutions computed by different processes; these are
separated from the language and semantics.
	
\item[$\star\star$] Issues of logic and control are separated to the point
that parallel PROLOG systems are usually built as speed-ups to 
$SLD$-resolution and have neither ``logic'' algorithms nor semantics of
their own.
For composition trees, they are implemented by
adding extra features to $SLD$-resolution. Specifically, composition
nodes are handled by binding arrays at compiler level.
\end{itemize}

In the previous sections, we have proposed coinductive trees (cf Definition \ref{df:coindt}),
as an alternative to composition trees. Coinductive trees serve as computational units in lazy (co)recursive derivations,
and therefore,  these coinductive tree transitions can be parallelised, as well.
For guarded logic programs, coinductive derivations allow for parallel and even non-deterministic implementations, as Sections \ref{sec:coindtree} and \ref{sec:impl} explain.
Here, we explain the two levels of parallelism in CoALP:  

\textbf{Level 1: Parallel construction of coinductive trees.}

Comparing coinductive
derivation trees with and-or parallel derivation trees, coinductive 
trees are more intrinsic: and-or parallel trees have mgu's built into 
a single tree, whereas mgu's are restricted to term-matching within the coinductive tree. Taking issues
of variable substitution from
the level of individual leaves to the level of trees affects
computations at least in two ways.  Parallel proof-search in branches
of a coinductive tree does not require synchronisation of
variables in different branches: they remain synchronised 
\emph{by construction} of the coinductive tree.   We illustrate with \texttt{ListNat}.

\begin{example}
  The coinductive trees of Figure \ref{pic:tree1} agree
  with the first part of the and-or parallel derivation tree for
  \texttt{list(cons(x,cons(y,x)))} in Figure \ref{pic:and-or}.  But the top left coinductive
  tree has leaves \texttt{nat(x)}, \texttt{nat(y)} and
  \texttt{list(x)}, whereas the and-or parallel derivation tree
  follows those nodes, using substitutions determined by
  mgu's. Moreover, those substitutions need not be consistent with
  each other: not only are there two ways to unify each of
  \texttt{nat(x)}, \texttt{nat(y)} and \texttt{list(x)}, but also
  there is no consistent substitution for \texttt{x} at all. In
  contrast, the coinductive trees handle such cases lazily.
\end{example}

Term-matching in coinductive trees permits the construction of every
branch in a coinductive tree independently of
the other branches. Moreover, for programs that are
guarded by constructors, such as \texttt{Stream} and \texttt{ListNat},
we avoid infinite branches or an infinite number of variables in a
single tree. 
Since both term-matching and guardedness are components of the 
``logic'' algorithm of coinductive derivation, the Kowalski's principle can be reformulated for CoALP as follows:

$$\emph{CoALP = Logic \textbf{is} Control}.$$

This distinguishes two approaches: \\
\emph{Parallel LP = and-or parallel derivation trees + explicit handling of parallel resources
   at compiler level};
		and\\
		\emph{CoALP = coinductive
  derivation trees + implicit  handling of parallel resources  ``by program construction''}.

	\textbf{Case-study: Coalgebraic logic programming and resources for parallelism.}\label{ap:B}

In this case study, our focus is on resource handling of parallelism
in logic programming.  We start by illustrating ground cases of
parallel derivations: these can be parallelised straightforwardly,
and coinductive trees and and-or parallel derivation trees
coincide.  We consider the inductive program \texttt{ListNat},
although a similar case-study could be done with a coinductive logic
program such as \texttt{Stream}.

\begin{example}
Consider the and-or parallel derivation tree for \texttt{ListNat} with
goal\\ $\texttt{list(cons(0,cons(0,nil)))}$ in Figure \ref{pic:ok}. 
\end{example}

	\begin{figure}
\begin{center}
\footnotesize{
  \begin{tikzpicture}[scale=0.2,baseline=(current bounding box.north),grow=down,level distance=20mm,sibling distance=130mm,font=\footnotesize]
  \node {$\mathtt{list(cons(0,cons(0,nil)))}$}
   child {[fill] circle (4pt)
              child { node {$\mathtt{nat(0)}$}
            child {[fill] circle (4pt)
               child { node {$\Box$}}}}
      child [sibling distance=90mm] { node {$\mathtt{list(cons(0,nil))}$}
        child {[fill] circle (4pt)
                   child { node {$\mathtt{nat(0)}$}
               child {[fill] circle (4pt)
                child { node {$\Box$}}}}
             child[sibling distance=90mm] { node{$\mathtt{list(nil)}$}
						child {[fill] circle (4pt)
                child { node {$\Box$}}}}}}}; 
  \end{tikzpicture}}
	\end{center}
\caption{\footnotesize{An and-or parallel derivation for the goal \texttt{list(cons(O,cons(O,nil)))}. }}
\label{pic:ok} 
\end{figure}
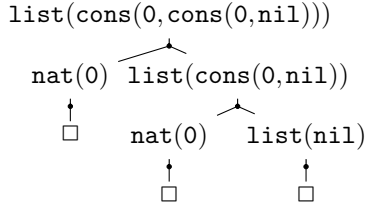


No additional syntactic annotations or variable synchronisation algorithms is required by CoALP when extending
from ground cases to the full fragment of first-order Horn logic with recursion and corecursion.
Not only  termination, but also soundness of parallelism will be guarded by program construction.

\begin{example}
Consider the coinductive
derivation for the goal $\mathtt{list(cons(x,cons(y,x)))}$ given in
Figure \ref{pic:tree1}. In contrast to the and-or parallel derivation
tree, and owing to the restriction of unification to term matching,
every coinductive tree in the derivation pursues fewer
variable substitutions than the corresponding and-or parallel
derivation tree does, cf Figure \ref{pic:and-or}. This allows one to
keep variables synchronised while pursuing parallel proof branches in
the tree. In particular, coinductive derivation of Figure \ref{pic:tree1} will report failure, as required for this example.
\end{example}

	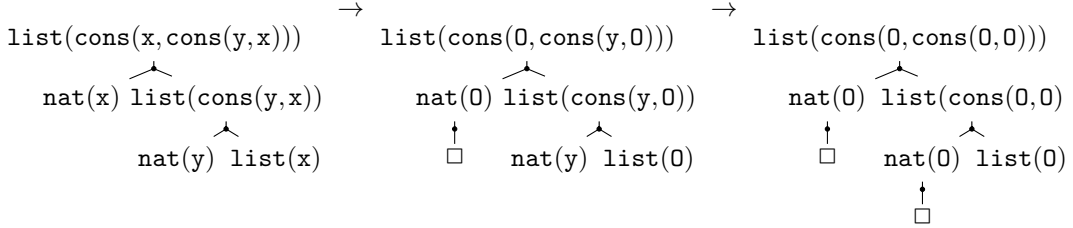
\begin{figure}
\begin{center}
\footnotesize{
  \begin{tikzpicture}[scale=0.2,baseline=(current bounding box.north),grow=down,level distance=20mm,sibling distance=95mm,font=\footnotesize]
  \node {$\mathtt{list(cons(x,cons(y,x)))}$}
   child {[fill] circle (4pt)
     child { node {$\mathtt{nat(x)}$}}
       child { node {$\mathtt{list(cons(y,x))}$}
        child [sibling distance=65mm] {[fill] circle (4pt)
       child { node{$\mathtt{nat(y)}$}}
         child { node{$\mathtt{list(x)}$}}}
       }}; 
  \end{tikzpicture}
$\rightarrow$
  \begin{tikzpicture}[scale=0.2,baseline=(current bounding box.north),grow=down,level distance=20mm,sibling distance=95mm,font=\footnotesize]
  \node {$\mathtt{list(cons(O,cons(y,O)))}$}
   child {[fill] circle (4pt)
              child { node {$\mathtt{nat(O)}$}
            child {[fill] circle (4pt)
               child { node {$\Box$}}}}
      child { node {$\mathtt{list(cons(y,O))}$}
        child [sibling distance=65mm] {[fill] circle (4pt)
                   child { node {$\mathtt{nat(y)}$}
               }
             child { node{$\mathtt{list(O)}$}}}}}; 
  \end{tikzpicture}
	$\rightarrow$
  \begin{tikzpicture}[scale=0.2,baseline=(current bounding box.north),grow=down,level distance=20mm,sibling distance=95mm,font=\footnotesize]
  \node {$\mathtt{list(cons(O,cons(O,O)))}$}
   child {[fill] circle (4pt)
              child { node {$\mathtt{nat(O)}$}
            child {[fill] circle (4pt)
               child { node {$\Box$}}}}
      child { node {$\mathtt{list(cons(O,O)}$}
        child [sibling distance=65mm]{[fill] circle (4pt)
                   child { node {$\mathtt{nat(O)}$}
               child [sibling distance=65mm] {[fill] circle (4pt)
                child { node {$\Box$}}}}
             child { node{$\mathtt{list(O)}$}}}}}; 
  \end{tikzpicture}
	}
\end{center}
\caption{\footnotesize{A coinductive derivation for the goal \texttt{list(cons(x,cons(y,x)))}.}}
\label{pic:tree1} 
\end{figure}


Note that coinductive trees not only permit to achieve
soundness where parallelism normally is not sound, but also they
achieve this without any kind of explicit resource handling.

\textbf{Level 2: Parallel transitions between coinductive trees.}
	
	Consider the leftmost coinductive tree of Figure \ref{pic:tree1}. It has three leaves with two distinct variables.
		Hence, three independent mgu's can be computed to unfold that tree; and the three tree transitions can be done in parallel.
		As the lazy nature of coinductive trees and guardedness checks of CoALP insure both soundness and termination of 
		computations at the level of each individual tree, this opens a possibility for parallel proof search
		through the state space of such trees. We discuss this in detail in Section \ref{sec:impl}.

	To conclude, CoALP gives a
different view of parallel resource handling:

\begin{enumerate}
\item We avoid explicit resource handling either at ``logic'' or
``control'' level; instead, we use implicit methods to control parallel
resources.

\begin{enumerate}
	\item In particular, we restrict unification to term matching:
in contrast to the inherently sequential unification
algorithm~\cite{DKM84}, it is parallelisable.  As a result, parallel
proof search in separate branches of a coinductive tree
does not require explicit synchronisation of variables.
\item Static guardedness checks of CoALP, introduced to guard corecursion, in fact insure that parallel scheduling of computations within the coinductive trees 
will never fall into a non-terminating thread; and parallel scheduling of coinductive tree transitions will never produce unsound results.
Again, this is achieved without introducing new syntax, just by the guarded program construction.
\end{enumerate}

\item The issues of logic and control are now bound together: coinductive
 trees provide both logic specification and resource
control.  Moreover, CoALP comes
with its own coalgebraic semantics that accounts for observational
behaviour of coinductive derivations.

\end{enumerate}

As the next section explains, this approach to parallelism can be viable and efficient.
See also \cite{KSH13} for a detailed study of CoALP's parallel features, in ground, Datalog, and full first-order case.

\section{Implementation}\label{sec:impl}

In \cite{KPS12}, we developed the first minimal prototype of CoALP in PROLOG, to show the feasibility of the coalgebraic logic programming approach, see  \emph{CoALP Prototype-1} in \cite{SK12}. However, it did not make use of parallelisation in modern computer architectures and was constrained by the mechanisms employed by the underlying PROLOG engine. 
Here, we present a new binary standalone implementation engineered using the \textbf{Go} programming language \cite{S12} available as   \emph{CoALP Prototype-2} in \cite{SK12}. Its most important new feature is the use of \textbf{Go}'s built in support for multithreading to achieve parallelisation by using \emph{goroutines} which are coroutines that can be executed in distinct threads. This new implementation also features two levels of parallelism (for coinductive trees and their transitions), static guardedness checks, and implicit handling of corecursion and parallelism. In this section, we describe the most important features arising in the implementation of CoALP.

\textbf{Construction of Coinductive trees} (cf. Definition \ref{df:coindt}) lies at the heart of CoALP's implementation. They are implemented by linking structural records (structs) which represent or-nodes and and-nodes through the use of arrays and pointers. And-nodes represent goal terms and contain a list of pointers to clauses that have heads which are still unifiable with the goal. An and-node with a list containing at least one such pointer is regarded as an open node. The root of any coinductive tree is an and-node constructed by the initial goal. 

\textbf{Guardedness} plays an important role in CoALP implementation, as Sections \ref{sec:corec} and \ref{sec:parallel} explain. For the proper operation of the CoALP algorithm, it needs to be ensured that a derivation step never produces an infinite and therefore non-well-founded coinductive tree.  This would block the search process by taking up infinite time to expand the tree. We have incorporated the \emph{Guardedness checks} in CoALP (cf Section \ref{sec:guard}); they are used to statically check the input programs, prior to the program run. Note that, in line with lazy corecursion in functional languages, while a coinductive tree may only be finite, the coinductive derivation may still be infinite (cf. \texttt{Stream} in Figure \ref{pic:stream}).

\textbf{Coinductive derivations} are transitions of coinductive trees. 
Whether the CoALP implementation is viewed as a sequential or parallel process, it can be described as follows. Construction of coinductive derivations for a given input program and goal is modeled as a uniform cost search through the graph of coinductive trees connected by the derivation operation. A derivation step here is constrained to first order unification of the first unifiable open node that has the lowest level in the tree; cf. Definition \ref{df:coind-res} and Figures \ref{pic:stream} and \ref{pic:tree1}. Other strategies, including non-deterministic methods are possible for selecting such open nodes; thereby determining substitutions for new coinductive tree transitions. 

\begin{example}
Looking at the $\texttt{ListNat}$ program from Example \ref{ex:listnat}, the tree with root\\ $\mathtt{list(cons(x, cons(y, x)))}$ is connected to $\mathtt{list(cons(0, cons(y, 0)))}$ by unification of the open node $\mathtt{nat(x)}$ with $\mathtt{nat(0)}$. This step is also shown in Figure \ref{pic:tree}. The following derivation and the resulting coinductive tree for $\mathtt{list(cons(0, cons(0, 0)))}$ contains no unifiable open nodes -- note that $\mathtt{list(0)}$ cannot be unified with any clause head of the input program. Only a very thin layer of sequential control in the implementation for this search is needed in the form of a priority search queue.
\end{example}

Using the substitution length of all the substitutions in the derivation chain as priority ranking, we gain an enumeration order even for a potentially infinite lazy derivation processes. Therefore, while an infinite number of coinductive trees can in principle be produced for the goal $\mathtt{list(x)}$, the algorithm returns $\mathtt{list(nil)}$, $\mathtt{list(cons(0,nil))}$ and then $\mathtt{list(cons(s(0),nil))}$ in a finite number of time-steps and keeps producing finite coinductive trees thereafter. Running CoALP \cite{SK12} for $\mathtt{list(x)}$, we get as output  the substitutions for the first three success trees:

\begin{itemize}
	\item[1] $\{x/nil\}$, 
	\item[3] $\{x/cons(x_1,y_1),x_1/0,y_1/nil\}$ and 
	\item[4] $\{x/cons(x_1,y_1),x_1/s(x_2),x_2/0,y_1/nil\}$.
\end{itemize}

Each possible coinductive tree will be produced after finite time, but since there may be infinitely many such trees, the coinductive derivations are implemented as lazy corecursive computations. Contrast this to PROLOG which produces the solutions $\mathtt{list(nil)}$ , $\mathtt{list(cons(0, nil))}$ , $\mathtt{list(cons(0, cons(0, nil)))}$, $\ldots$ but never $\mathtt{list(cons(s(0),nil))}$ for the $\texttt{ListNat}$ program and goal $\mathtt{list(x)}$. Thereby, it does not generate the same set of solutions even if run indefinitely and does not discover some of the solutions that CoALP does.

\textbf{A new approach to Backtracking} is taken, as CoALP explores simultaneously several and-or-choices in a coinductive tree. In contrast to PROLOG, no trail stack is maintained and no backtracking (in the classical sense of \cite{Llo88}) is needed. If a coinductive tree has no open unifiable nodes, it will simply be discarded. If alternative mgu's existed during the derivation steps, they open up different branches in coinductive derivations. Therefore, CoALP implicitly represents alternative mgu's by coinductive trees in the priority search queue. The only time variable bindings may be undone is when checking for unifiability of terms during the derivation step. However, this is only done on copies of the original terms to ensure thread safety and to avoid unnecessary locks and therefore sequential barriers. Furthermore, this is done locally and does not characterize or regulate the overall global search flow.

\textbf{Parallelisation of coinductive trees.} Given that no infinite derivation tree can be generated by a guarded program, the CoALP approach provides multiple points where parallelisation takes place, while still enumerating every possible coinductive success subtree. The use of term matching to traverse and expand trees allows for parallelisation of work without explicit variable synchronization while operating directly on a single tree. 

However, if the coinductive trees are small or few open nodes exist, such as in the running examples \texttt{Stream} and \texttt{ListNat}, the setup and initial communication overhead between parallel threads that process the tree does not usually offset speedup that can be achieved. Therefore, it is dynamically decided during execution whether a program generates sufficiently complex coinductive trees to warrant this parallelisation strategy. Future research will focus on efficient heuristics to decide how this trade-off should be made. 

Term matching can be performed in parallel, but if the terms are small, no practical speedup will be obtained when working with multiple threads. In such cases, it is more efficient to perform distinct term matching operations in parallel by dispatching work on multiple coinductive trees in parallel.

Ground logic programs do not need transitions between the coinductive trees to complete the computation.
 Logic programs containing variables but no function symbols of arity $n>0$ can all be soundly translated into finitely-presented ground logic programs. The most famous example of such a language is Datalog~\cite{UllmanG88,Kanellakis88}.
The advantages of Datalog are easier implementations and a greater capacity for parallelisation. 

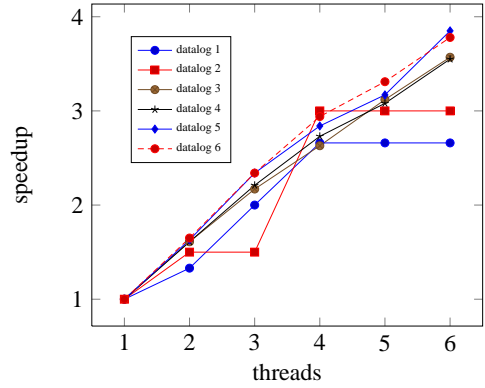
\begin{figure}
\centering
\begin{tikzpicture}[scale=.75]
	\begin{axis}[xlabel= threads,
		ylabel=speedup,legend style={at={(0.1,.9)},anchor=north west}]
	\addplot  coordinates {
		(1,1)
		(2,1.33)
		(3,2)
		(4,2.66)
		(5,2.66)
		(6,2.66)
	
	};
	\addplot coordinates {
		(1,1)
		(2,1.5)
		(3,1.5)
		(4,3)
		(5,3)
		(6,3)
	
	};
	\addplot coordinates {
		(1,1)
		(2,1.61)
		(3,2.17)
		(4,2.63)
		(5,3.12)
		(6,3.57)

	};
	\addplot  coordinates {
		(1,1)
		(2,1.61)
		(3,2.21)
		(4,2.73)
		(5,3.08)
		(6,3.55)
	
	};
	\addplot  coordinates {
		(1,1)
		(2,1.63)
		(3,2.34)
		(4,2.84)
		(5,3.17)
		(6,3.85)
	
	};
        \addplot  coordinates {
		(1,1)
		(2,1.65)
		(3,2.34)
		(4,2.94)
		(5,3.31)
		(6,3.78)
	
	};

	\legend{{\tiny datalog 1}, {\tiny datalog 2}, {\tiny datalog 3}, {\tiny datalog 4}, {\tiny datalog 5}, {\tiny datalog 6}}

	\end{axis}
\end{tikzpicture}
\caption{Speedup of Datalog programs, relative to the base case with 1 thread,  with different number of threads expanding the derivation tree.}\label{fig:datalog}
\end{figure}

Figure \ref{fig:datalog} shows the speedup that can be gained by constructing and-or parallel trees for Datalog programs in our system.
The Datalog programs are randomly generated and can be examined in~\cite{SK12}. As can be seen in Figure~\ref{fig:datalog}, the speedup is
 significant and scales with the number of threads.

\textbf{Parallelisation of coinductive derivations} is more efficient than parallelisation within one coinductive tree for programs like \texttt{ListNat} and \texttt{Stream}. On the search queue level of the algorithm, multiple trees that still have open nodes and possible derivations are dispatched to one or more worker threads. They perform the coinductive derivation steps in parallel. To keep communication minimal, the coinductive trees are compacted by e.g. pruning closed leaves and shortening chains that have no branches in the tree. Since expanding and checking coinductive trees does not always take the same amount of time for each tree, some worker threads might return results earlier than others and thereby disrupt the enumeration order. So, we do not allow them to show results immediately and  directly to the user. CoALP guarantees that success trees which are enumerated sequentially will also be found when working in multithreaded context albeit maybe later. Returning results in the enumeration order of substitution lengths to the user can still be achieved by a little more sequential overhead. For example, the user can specify the option to buffer and sort success coinductive trees until it is guaranteed that no lower order coinductive trees are still being processed or are in the priority search queue.

Considering the other direction of reducing sequential overhead in maintaining the search queue, there is the possibility of using complementary enumeration schemes and thereby partition the search queue into smaller queues that each worker thread maintains on its own. However, this may shift the order of solutions since some worker threads may enumerate only solutions that are computationally easier to find. Thereby a trade-off is to be made between maintaining a perfect ordering or faster processing of coinductive trees. At any rate, the derivations remain sound by the program guardedness and coinductive tree construction, cf Sections  \ref{sec:corec} and \ref{sec:parallel}; and this allows for a range of experiments
on parallelisation for the future.

\section{Conclusions and Future Work}\label{sec:concl}

The main feature of the coalgebraic logic programming approach is its
generality: it is suitable for both inductive and coinductive logic
programs, for programs with variable dependencies or not, and for
programs that are unification-parallelisable or inherently
sequential. Many distinctions that led to a variety of engineering
solutions in the design of corecursive and concurrent logic programs
\cite{GPACH12,GuptaBMSM07,SimonBMG07} are erased here, with
resource-handling delegated to a logic algorithm; and issues
of logic and control, semantics and execution, become inseparable.

The original contributions of this paper relative to the earlier papers~\cite{KMP10,KP11,KP11-2} are the Coalgebraic Calculus of Infinite trees (Section~\ref{sec:trees}), 
operational semantics for non-deterministic derivations (Section~\ref{sec:TO}), extended Guardedness conditions for CoALP (Section~\ref{sec:guard}), and Parallel and Corecursive Implementation of CoALP
in Go (Section~\ref{sec:impl}). Additionally, the paper develops a unfying theory and notation for paralelleism and corecursion in logic programming, putting a new perspective on 
 earlier results~\cite{KMP10,KP11,KP11-2}. All proofs appear here for the first time. 

In future, we plan to investigate the integration of coalgebraic logic
programming with methods of resource handling in state-of-the-art
coinductive logic programming ~\cite{GPACH12,GuptaBMSM07,SimonBMG07},
as well as in modern concurrent logic programming
systems~\cite{GPACH12}. Furthermore, we would like to investigate
whether coalgebraic logic programming has potential to play a positive
role in type inference, cf. \cite{AnconaLZ08}.

The analysis of this paper can be extended to more expressive logic
programming languages, such as~\cite{Girard87,HodasM94,Pym,MillerN86}, also to
functional programming languages in the style of
\cite{PaulsonS89,AnconaLZ08}. We deliberately chose our running
examples to correspond to definitions of inductive or coinductive
types in such languages. 

The key fact driving our analysis has been the observation that the
implication $\gets$ acts at a meta-level, like a sequent rather than a
logical connective.  That observation extends to first-order fragments
of linear logic and the Logic of Bunched Implications
\cite{Girard87,Pym}. So we plan to extend the work in the paper to
logic programming languages based on such logics.

The situation regarding higher-order logic programming languages such
as {\em $\lambda$-PROLOG} \cite{MillerN86} is more subtle. Despite
their higher-order nature, such logic programming languages typically
make fundamental use of sequents. So it may well be fruitful to
consider modelling them in terms of coalgebra too, albeit probably on
a sophisticated base category such as a category of Heyting algebras.

\bibliographystyle{abbrv}


\end{document}